\documentclass[floats,12pt,a4paper,onecolumn]{article}
\usepackage{amsfonts}
\usepackage{amsmath}
\usepackage{calc}
\usepackage{subfigure}
\usepackage{graphicx} 
\usepackage{enumerate}
\usepackage{hyperref}
\usepackage{amssymb}
\usepackage{amsthm}
\usepackage{todonotes}

\newcommand{\appref}[1]{\hyperref[#1]{{Appendix~\ref*{#1}}}}
\newcommand{\be}{\begin{eqnarray} \begin{aligned}}
\newcommand{\ee}{\end{aligned} \end{eqnarray} }
\newcommand{\benn}{\begin{eqnarray*} \begin{aligned}}
\newcommand{\eenn}{\end{aligned} \end{eqnarray*}}
\newcommand{\cancel}[1]{} 
\usepackage[normalem]{ulem}

\newcommand*{\cB}{\mathcal{B}}
\newcommand*{\cC}{\mathcal{C}}
\newcommand*{\cE}{\mathcal{E}}

\newcommand*{\cH}{\mathcal{H}}

\newcommand*{\cL}{\mathcal{L}}

\newcommand*{\cN}{\mathcal{N}}

\newcommand*{\cQ}{\mathcal{Q}}
\newcommand*{\cR}{\mathcal{R}}

\newcommand*{\cP}{\mathcal{P}}
\newcommand*{\cS}{\mathcal{S}}
\newcommand*{\cU}{\mathcal{U}}
\newcommand*{\cT}{\mathcal{T}}
\newcommand*{\cV}{\mathcal{V}}

\newcommand*{\cX}{\mathcal{X}}
\newcommand*{\cY}{\mathcal{Y}}

\newcommand{\bc}{\begin{center}}
\newcommand{\ec}{\end{center}}

\newcommand{\id}{\mathbb{I}}

\newcommand{\e}{\mathrm{e}}

\newtheorem{theorem}{Theorem}[section]

\newtheorem{lemma}[theorem]{Lemma}

\newtheorem{corollary}[theorem]{Corollary}

\newcommand{\argmax}{\mathop{\mathrm{argmax}}\nolimits}

\usepackage{mathtools}
\mathtoolsset{showonlyrefs}

\usepackage{amsfonts}

\def\id{\mathbb{I}}

\def\01{\{0,1\}}

\newcommand{\ket}[1]{|#1\rangle}

\newcommand{\proj}[1]{|#1\rangle\langle#1|}

\newcommand{\comment}[1]{}

\usepackage{tikz}
\input{pgfutil-common}
\usetikzlibrary{quantikz}
\usetikzlibrary{backgrounds,fit,decorations.pathreplacing,calc,positioning,shapes.geometric,shapes.symbols,shapes.misc}                 
\tikzset{
    operation/.style={
        draw,
        rectangle,
        minimum height=0.8cm
    },
    idealdecoder/.style={
        draw,
        isosceles triangle,
        isosceles triangle apex angle=90,
        shape border rotate=0
    },
    semicirc/.style args={#1,#2}{
    	semicircle,
    	minimum width=#1,
    	draw,
    	anchor=arc end,
    	rotate=#2
    },
    circ/.style={
    	circle,
    	minimum width=0.8cm,
    	draw
    },
    phasemod/.style={shape=rectangle,minimum size=3pt,thick,draw},
    arrowsmod/.style={thick,color=white},
}
\newcommand{\cwmod}{\ifthenelse{\the\pgfmatrixcurrentcolumn>1}{\arrow[arrowsmod,yshift=0.05cm]{l}\arrow[arrowsmod,yshift=-0.05cm]{l}}{}}
\newcommand{\vcwmod}[1]{
	\edef\start{\the\pgfmatrixcurrentrow-\the\pgfmatrixcurrentcolumn}
	\edef\end{\the\numexpr#1+\pgfmatrixcurrentrow\relax-\the\pgfmatrixcurrentcolumn}
	\expandafter\expandafter\expandafter\vcwexplicit\expandafter\expandafter\expandafter{\expandafter\start\expandafter}\expandafter{\end}
}
\newcommand{\cwbendmod}[1]{
	\vcw{#1}\cw
	\edef\cell{\the\pgfmatrixcurrentrow-\the\pgfmatrixcurrentcolumn}
	\expandafter\pgfutil@g@addto@macro\expandafter\tikzcd@atendlabels\expandafter{%
		\expandafter\latephasemod@end\expandafter{\cell}
	}
}
\newcommand{\cwbendmodmod}[1]{
	\vcw{#1}\cwmod
	\edef\cell{\the\pgfmatrixcurrentrow-\the\pgfmatrixcurrentcolumn}
	\expandafter\pgfutil@g@addto@macro\expandafter\tikzcd@atendlabels\expandafter{%
		\expandafter\latephasemod@end\expandafter{\cell}
	}
}
\newcommand{\latephasemod@end}[1]{
	\node [phasemod,inner sep=1.5pt] at (\tikzcdmatrixname-#1) {};
}

\newcommand*{\gkp}{\mathsf{GKP}}

\newcommand*{\pA}{\Pi_A}
\newcommand*{\pB}{\Pi_B}

\title{Oscillator-to-oscillator codes do not have a threshold}
\usepackage{authblk}

\begin{document}
\author{Lisa H\"anggli}
 \author{Robert K\"onig}
 \affil{Zentrum Mathematik, Technical University of Munich, Germany}
\maketitle

\begin{abstract}
It is known that continuous variable quantum information cannot be protected against naturally occurring noise using Gaussian states and operations only. Noh et al.~(PRL 125:080503, 2020) proposed bosonic oscillator-to-oscillator codes 
 relying on  non-Gaussian resource states as an alternative, and showed that these encodings can lead to a reduction of the effective error strength at the logical level as measured by the variance of the classical  displacement noise channel. An oscillator-to-oscillator code  embeds $K$~logical bosonic modes (in an arbitrary state) into $N$~physical modes by means of a Gaussian $N$-mode unitary and  $N-K$~auxiliary one-mode Gottesman-Kitaev-Preskill-states.

Here we ask if -- in analogy to qubit error-correcting codes -- there are families of oscillator-to-oscillator codes with the following threshold property: They allow to convert physical displacement noise with variance below some threshold value to logical  noise with variance upper bounded by any (arbitrary) constant. We find that this is not the case if 
encoding unitaries involving a constant amount of squeezing and maximum likelihood error decoding  are used. We show a general lower bound on the logical error probability which is only a function of the amount of squeezing and independent of the number of modes. As a consequence,  any physically implementable family of oscillator-to-oscillator codes 
 combined with  maximum likelihood error decoding does not admit a threshold. 
\end{abstract}

\section{Introduction}
The theory of fault-tolerance builds on the fact that information can be protected by introducing redudancy combined with suitable recovery procedures. A prime example is the classical  repetition code encoding one logical bit into three physical bits according to
\begin{align}
0\mapsto 000\qquad\textrm{ and }\qquad 1\mapsto 111\ .
\end{align}  Assuming that each physical bit undergoes independent bit flip noise with probability~$\epsilon$, this encoding improves protection as long as $\epsilon<1/2$: The probability of erroneous correction is of order 
\begin{align}
\epsilon_L=O(\epsilon^2)\ .
\end{align}
Higher-order error suppression of the form~$\epsilon_L=O(\epsilon^N)$ can be achieved by the use of error-correcting codes with larger distance, obtained e.g., by concatenation. Such error supression mechanisms are a cornerstone of John von Neumann's pioneering theory of fault-tolerant classical computation presented in a paper titled ``Probabilistic Logics and the Synthesis of Reliable Organisms from Unreliable Components''~\cite{vonNeumann56}. For present-day efforts to synthesize  reliable quantum computers it is imperative to identify similar error suppression methods for quantum information in its various forms. Great strides have been made in this direction. For encoding qubits into qubits, the theory of stabilizer codes is highly developed. On the other hand, when encoding qubits into bosonic modes,  Gottesman-Kitaev-Preskill (GKP) codes~\cite{gkp01}, cat codes~\cite{crochaneetal99,leghtasetal13}, and binomial codes~\cite{albertetal16}
are prominent candidates for experimental realizations. 

More recently, Noh, Girvin, and Jiang~\cite{NohGirvinJiang20arX,NohGirvinJiang20} proposed new encodings of some number~$K$ of bosonic modes into $N>K$ bosonic modes along with suitable recovery maps. Such oscillator-to-oscillator codes can indeed achieve error suppression, as exemplified by the so-called  GKP two-mode squeezing code for which $(N,K)=(2,1)$: In~\cite{NohGirvinJiang20}, it is shown that physical-level phase space displacement noise (so-called classical noise) of standard deviation~$\sigma$ is reduced to  displacement noise with standard deviation 
\begin{align}
\sigma_L=O(\sigma^2\cdot polylog(1/\sigma))\label{eq:errorsuppressionosctoosc}
\end{align} at the logical level provided that the initial noise strength~$\sigma$ is below some constant.  Furthermore, higher-level suppression  can be achieved using $N>2$ physical modes and so-called GKP-squeezed repetition codes, as argued in~\cite{NohGirvinJiang20arX}. This achievement is based on a non-trivial use of non-Gaussian resources: It is well-known that  purely Gaussian fault-tolerance operations cannot provide such error suppression~\cite{niseketalnogoQerror09,vuillotetal}, since even Gaussian errors (such as the classical noise channel) cannot be corrected by means of Gaussian operations. The constructions introduced by Noh et al.~use GKP states as a non-Gaussian resource to circumvent these no-go results.

In the schemes considered in~\cite{NohGirvinJiang20arX,NohGirvinJiang20},  a $K$-mode bosonic state~$\ket{\Psi}$ is encoded with the map 
 \begin{align}
 \ket{\Psi}\mapsto U^{(N)}(\ket{\Psi}\otimes \ket{\gkp}^{\otimes N-K})\label{eq:GKPencodingmap}
 \end{align}
 by means of an $N$-mode Gaussian unitary~$U^{(N)}$ and $N-K$~copies of the canonical one-mode GKP state~$\ket{\gkp}$. The class of these codes will be referred to as oscillator-to-oscillator codes of GKP type here, or simply oscillator-to-oscillator codes. For the GKP two-mode squeezing code, the unitary~$U^{(2)}$ is a two-mode squeezing operation, whereas for the GKP-squeezed repetition code, a recursively defined unitary~$U^{(N)}_{rep}$ is used.

 Here we investigate to what extent the error suppression achieved by oscillator-to-oscillator codes can be exploited towards establishing a fault-tolerance threshold theorem for quantum computation with bosonic modes. 
Roughly, a threshold theorem amounts to the claim that for sufficiently weak noise (i.e., noise strength below some threshold value), errors can be recovered from with any desired accuracy at the cost of introducing a sufficient amount of redundancy. In order to 
formalize and motivate this problem, it is useful to review the case of qubits: Here  a rigorous threshold theorem for fault-tolerant quantum computation has been established in  a series of pioneering works~\cite{shor96faulttolerance,knilllaflammezurek96arX,knilllaflammezurek98resilientqc,knilllaflammezurek98resilientqc2,kitaev97,aharonov99}.   

\paragraph{Fault-tolerance threshold with qubit error correcting codes.}
For simplicity, let us  focus on the question of fault-tolerant memories.   To be concrete, let us assume that each physical qubit is  affected by a probabilistic Pauli noise channel~$\cN:\cB(\mathbb{C}^2)\rightarrow\cB(\mathbb{C}^2)$. The strength~$\delta$ of the noise can be measured e.g., by the diamond norm distance~$\delta=\|\cN-\mathsf{id}_{\cB(\mathbb{C}^2)}\|_{\diamond}$ of~$\cN$ to the single-qubit identity channel~$\mathsf{id}_{\cB(\mathbb{C}^2)}$. A quantum code encoding $K$~logical qubits into $N$~physical qubits is a $2^K$-dimensional subspace~$\cC_N\subset(\mathbb{C}^2)^{\otimes N}$. It can be understood as the image of an isometric encoding map from~$\mathbb{C}^2$ to $(\mathbb{C}^2)^{\otimes N}$ which takes the form~\eqref{eq:GKPencodingmap} for a qubit state $\ket{\Psi}$, an $N$-qubit encoding unitary~$U^{(N)}$, and  $\ket{\gkp}$ replaced by the computational basis state~$\ket{0}$. An associated recovery operation is a quantum channel~$\cR_N:\cB\left((\mathbb{C}^2)^{\otimes N}\right)\rightarrow\cB\left((\mathbb{C}^2)^{\otimes N}\right)$, typically composed of syndrome measurement and a  subsequent (conditional) unitary correction operation. 
We say that the pair $(\cC_N,\cR_N)$ recovers with accuracy~$\epsilon>0$ from noise of strength~$\delta$ if  for any $1$-qubit noise channel $\cN$ with strength~$\delta$, the composition~$\cR_N\circ\cN^{\otimes N}$ takes the set of states supported on~$\cC_N$ to itself and satisfies $\|\left.\cR_N\circ\cN^{\otimes N}\right|_{\cB(\cC_N)}-\mathsf{id}_{\cB(\cC_N)}\|_{\diamond}\leq \epsilon$.

With these definitions, the concept of a threshold can be formulated as follows: An (infinite) family $\{(\cC_N,\cR_N)\}_{N\in \Gamma}$ of codes (where $\cC_N\subset (\mathbb{C}^2)^{\otimes N}$) with associated recovery maps has a threshold if there is some value $\delta_0>0$ such that for any error strength~$\delta<\delta_0$, and any accuracy~$\epsilon>0$, there is some $N=N(\delta,\epsilon)\in\Gamma$ such that $(\cC_N,\cR_N)$ recovers with accuracy at least~$\epsilon$ from noise of strength~$\delta$. This property formalizes the notion that essentially ideal quantum memories can be emulated using a sufficient number of noisy qubits, provided that these are not too noisy (i.e., have a noise strength below some threshold). We note that in contrast to typical settings in   fault-tolerant quantum computing, where any physical operation is assumed to be imperfect, here  we consider a simplified notion of quantum memories by assuming that the recovery operation can be implemented perfectly. In fact, this assumption strenghtens the no-go result we find.

One of the crowning achievements of the theory of quantum error correction was to show that quantum code families for qubits exhibiting a threshold exist. In fact, not only do they exist, but we have explicit constructions at hand. Indeed, all that is needed is a family of  quantum error-correcting codes with a code distance scaling extensively with the number~$N$ of physical qubits. Such a code achieves error suppression which is exponential in~$N$. By choosing $N$~to be sufficiently large, this means  that any (sufficiently weak) preexisting physical noise can be converted to  noise below any desired strength at the logical level. Code families with macroscopic distance can be obtained by a variety of construction techniques.

To deploy such codes in the real world,  both the threshold value~$\delta_0$ and the actual dependence of $N=N(\delta,\epsilon)$ on the parameters~$\delta$ and $\epsilon$ are of key interest: They determine the required resources. For typical codes and decoders considered in the literature, $N$ scales as $N(\delta,\epsilon)=poly(\log 1/\delta,\log 1/\epsilon)$. 
The parameter~$N$ also determines the required amount of computational resources: For most code families (in particular stabilizer codes), the encoding unitary~$U^{(N)}$ involves $poly(N)$ elementary gates (e.g., one- and two-qubit gates). Similarly, for practical relevance, we typically are interested in recovery operations~$\cR_N$ that require~$poly(N)$ measurements and $poly(N)$-time classical computation. In summary, we have qubit codes that allow for efficient encoding and recovery/decoding as measured in terms of basic operations (one- and two-qubit gates, one-qubit measurements, and classical arithmetic). Most importantly, the nature of these operations does not change with increasing~$N$. In this sense, quantum fault-tolerance with qubit codes can  be considered strictly as a way of adding redudancy: Simply using more of the same (noisy) resources yields the desired effect.

\paragraph{No fault-tolerance threshold for oscillator-to-oscillator codes.} 
Motivated by the favorable properties of qubit quantum error-correcting codes, we ask if an analogous fault-tolerance threshold statement holds for oscillator-to-oscillator codes. A candidate formulation is immediately obtained by substituting single-qubit noise by a bosonic noise channel $\cN:\cB(L^2(\mathbb{R}))\rightarrow\cB(L^2(\mathbb{R}))$, which we will assume to be a classical noise channel with variance~$\sigma^2$. The question is then whether there is a family of subspaces~$\cC_N\subset L^2(\mathbb{R})^{\otimes N}$ that are images of encoding maps as in~\eqref{eq:GKPencodingmap}, each defined by a Gaussian unitary~$U^{(N)}$ on $N$~modes, together with suitable recovery operations~$\cR_N:\cB(L^2(\mathbb{R})^{\otimes N})\rightarrow\cB(L^2(\mathbb{R})^{\otimes N})$, such that the following holds: Provided that $\sigma<\sigma_0$ for some constant threshold~$\sigma_0$, there is some $N=N(\sigma,\sigma')$ for any $\sigma'>0$ such that the composition~$\cR_N\circ \cN^{\otimes N}$ is given by displacement noise of variance less than~$\sigma'^2$ when restricted to the code space~$\cC_N\cong L^2(\mathbb{R})^{\otimes K}$. In view of  error suppression properties such as~\eqref{eq:errorsuppressionosctoosc} achieved by certain oscillator-to-oscillator codes, such a fault-tolerance property may appear to be potentially feasible using these codes. A corresponding result would then suggest that the idea that ``more is better'' when it comes to fault-tolerance extends to the setting of oscillator-to-oscillator codes.

We show here that this is not the case if maximum likelihood error decoding is used: no family of oscillator-to-oscillator codes with physically meaningful encoding unitaries exhibits a threshold in the above sense. The key difference to the finite-dimensional setting is that the feasibility of implementing a bosonic unitary~$U$ does not depend on measures such as  its gate complexity only: In addition, the amount of squeezing it introduces is a key quantity for experimental purposes. This quantity, expressed by a ``squeezing measure'' $\mathsf{sq}(U)$ we introduce in Section~\ref{sec:quantifyingsqueezing}, should realistically be bounded by a constant.

For a two-mode squeezing unitary~$U^{(2)}$ as used in the GKP two-mode squeezing code, the squeezing measure~$\mathsf{sq}(U^{(2)})$ is in one-to-one correspondence with the so-called gain~$g(U^{(2)})$ more commonly used in physics.
In~\cite{NohGirvinJiang20}, the authors argue that  the error suppression~\eqref{eq:errorsuppressionosctoosc} results for a gain which scales as~$g^*\sim O((\sigma^2\log(1/\sigma))^{-1})+1/2$. The fact that this diverges as $\sigma\rightarrow 0$  gives a first indication that at least using the GKP two-mode squeezing code only (and possibly some kind of concatenation) may be problematic when trying to show a threshold. On the other hand, for the $N$-mode GKP-squeezed repetition code proposed in~\cite{NohGirvinJiang20arX}, the amount of squeezing~$\mathsf{sq}(U^{(N)}_{rep})$ of the encoding unitary~$U^{(N)}_{rep}$ diverges as $N\rightarrow\infty$. Thus this code family is not a candidate for giving a threshold theorem if we restrict our attention to constant-squeezing encoding unitaries.

Our no-go-result is not restricted to the $N$-mode repetition code only, however. Instead, we  show that for {\em any} family of $N$-mode Gaussian unitaries $\{U^{(N)}\}_{N}$ with the property that~$\mathsf{sq}(U^{(N)})$ is bounded by a constant independent of~$N$, the associated family of oscillator-to-oscillator codes does not exhibit a threshold. In particular, simply using more modes does not solve the fault-tolerance threshold issue unless increasingly stronger squeezing is used.

As indicated above, this result is established for the maximum likelihood error decoder based on GKP-type syndrome information. The strategy of this decoder is to identify the most likely error yielding the observed syndrome, and to subsequently correct for this error. It performs at least as well as the decoder used in~\cite{NohGirvinJiang20arX,NohGirvinJiang20} -- the latter is a heuristic approximation of maximum likelihood error decoding. While maximum likelihood error decoding is the natural choice when aiming to identify the physical error applied to the system, 
our work does not immediately provide insight on the performance of other decoding strategies.
 In particular, it does not exclude the possibility that maximum likelihood (error coset) decoding may outperform the decoder considered here. Such a decoder proceeds by identifying the most likely coset of errors (modulo the stabilizer group). This involves coset probabilities that typically cannot be evaluated in closed form. 
 
 \subsubsection*{Outline}
  
The remainder of the paper is organized as follows: We first introduce oscillator-to-oscillator codes and the physical setup considered here in  detail in Section~\ref{sec:oscillatortooscillatorcodes}. In Section~\ref{sec:problemreformulation}, we reformulate the problem of analysing
the ability of oscillator-to-oscillator codes to recover from errors as a purely classical estimation problem. In the same section, we also present our no-go result for oscillator-to-oscillator codes. The underlying  technical results concerning the classical estimation problem are established  in Section~\ref{sec:classicalanalysis}.

\section{Oscillator-to-oscillator codes}\label{sec:oscillatortooscillatorcodes}
In this section we introduce oscillator-to-oscillator codes and the necessary corresponding background. In particular, we explain how squeezing in Gaussian unitaries can be quantified, and specify the noise channel and  recovery map  considered. We also introduce the figure of merit for the proof of our no-go theorem for oscillator-to-oscillator codes: the decoding success probability.
    
\subsection{Quantifying squeezing in $N$-mode Gaussian unitaries}\label{sec:quantifyingsqueezing}
  An $N$-mode  bosonic system is described by $N$~pairs $(Q_j,P_j)$, $j=1,\ldots,N$, of position and momentum operators associated with the quadratures of the electromagnetic field modes. The canonical commutation relations can be compactly described by gathering these mode operators in a tuple~$R\coloneqq (Q_1,P_1,\dots, Q_N,P_N)$: They take the form
  \begin{align}
    [R_j,R_k]=iJ_{jk}\id\qquad\textrm{ where }\qquad
J=\bigoplus_{j=1}^N\begin{pmatrix}
0&1\\-1&0
\end{pmatrix}\ .
  \end{align}
In its most general meaning, the term Gaussian unitary refers to a unitary generated by a Hamiltonian which is at most quadratic in the mode operators. Here we use the term Gaussian unitary exclusively for those generated by purely quadratic Hamiltonians, and treat linear contributions and the corresponding unitaries separately in Section~\ref{subsec:displacements}. 
The action of a Gaussian unitary~$U$ is completely determined by an element $S\in  Sp(2N,\mathbb{R})$ of the symplectic linear group $ Sp(2N,\mathbb{R})\coloneqq \left\{S\in\mathrm{Mat}_{2N\times2N}(\mathbb{R})\ |\ S^TJS=J\right\}$: When conjugated by the unitary $U=U_S$, the  quadrature operators transform linearly as 
  \begin{align}
U_SR_jU_S^{\dagger}=\sum_{k=1}^{2N}S_{j,k}R_k\ ,\quad j=1,\dots,2N\ .   \label{eq:quadraticoperations}
  \end{align}

  A special role is played by the harmonic oscillator Hamiltonian $H_0=\sum_{j=1}^N (Q_j^2+P_j^2)$, which determines the energy of a state and generates rotations in phase space. A unitary~$U$ is called passive if it commutes with~$H_0$, i.e., if it preserves the energy. Passive Gaussian unitaries are particularly easy to implement: every such unitary is a composition of phaseshifters and beamsplitters, see~\cite{reckzeilingeretal94}. The latter belong to the set of operations referred to as passive linear optics, and are typically readily available. The set of passive Gaussian unitaries is in one-to-one correspondence (via the map $S\mapsto U_S$) with the group~$K(N)= Sp(2N,\mathbb{R})\cap O(2N,\mathbb{R})$ of orthogonal symplectic matrices (here $O(2N,\mathbb{R})=\left\{O\in\mathrm{Mat}_{2N\times2N}(\mathbb{R})\ |\ O^TO=I\right\}$), as the orthogonality constraint ensures that~$H_0$ is left invariant under conjugation.

  In contrast to passive Gaussian unitaries, active ones are significantly more challenging to implement in general. Such evolutions involve squeezing, i.e., a process where e.g., photon pairs are created or annihilated. Realizing such dynamics requires elements from non-linear quantum optics such as birefringent materials. An example is the one-mode squeezing unitary~$U_z$ associated with the symplectic matrix $\mathsf{diag}(z,1/z)$, $z\in \mathbb{R}\backslash\{1\}$. It reduces the quantum noise associated with one quadrature while increasing that of the other. This is particularly challenging to implement for large~$z$: With present technology, only squeezing of roughly~$10 dB$ (corresponding to $z\sim \sqrt{10}$) is feasible~\cite{tendbsqueezing}.

  In fact, one-mode squeezing unitaries together with phaseshifters and beamsplitters generate the set of all Gaussian unitaries according to a certain normal form for symplectic matrices: For any $S\in Sp(2N)$ there are $O_1,O_2\in K(N)$ such that
  \begin{align}
S&=O_1 ZO_2\ ,\quad\textrm{ with }\quad Z=\mathsf{diag}(z_1,1/z_1,\ldots,z_n,1/z_n) \quad \textrm{ where } \quad z_1,\ldots,z_n\in(0,\infty)\ , \label{eq:eulerdecomposition}
  \end{align}
  according to the Euler decomposition (see e.g.,~\cite{arvindtutta}).  In particular, the associated Gaussian unitary~$U_S$ can be implemented by realizing the unitaries corresponding to $O_1$ and $O_2$ (using passive linear optics only), and by applying the one-mode-squeezing unitary $U_{z_j}$ to every mode $j=1,\ldots,N$. 

  Note that for any passive Gaussian unitary $U=U_S$, we clearly have $z_j=1$ for all $j=1,\ldots,N$ in the decomposition~\eqref{eq:eulerdecomposition}. Moreover, we can also use Eq.~\eqref{eq:eulerdecomposition} more generally (going beyond passive unitaries) to quantify the amount of squeezing of any Gaussian unitary $U_S$ by the squeezing measure
  \begin{align}
\mathsf{sq}(U_S)\coloneqq \max_{j=1,\ldots,N} \{z_j,1/z_j\}\ .
  \end{align}
  For later use, we note that Eq.~\eqref{eq:eulerdecomposition} immediately implies that the matrix $S^TS$ is positive definite and
  \begin{align}
\mathsf{sq}(U_S)&=\sqrt{\lambda_{\max}\left(S^TS\right)}\ ,\label{eq:squeezingmeasureU}
  \end{align}
  where $\lambda_{\max}(M)$ denotes the maximal eigenvalue of a positive definite matrix~$M$.   The quantity~$\mathsf{sq}(U)$ describes the maximal amount of squeezing required for any single mode when implementing the unitary~$U$.

  \subsection{Displacement noise\label{subsec:displacements}}
  Hamiltonians that are linear combinations of the mode operators generate phase space displacements.
  The corresponding unitaries are parameterized by  vectors~$\xi\in\mathbb{R}^{2N}$: For $\xi\in\mathbb{R}^{2N}$, we denote by $D(\xi)$ the (unitary) displacement operator 
\begin{align}
D(\xi)\coloneqq \e^{i\xi^TJR}\ .
\end{align}
It transforms the mode operators as
\begin{align}
D(\xi) R_j D(\xi)^\dagger &=R_j+\xi_j\id\qquad\textrm{ for }\qquad j=1,\ldots,2N\ .
\end{align}
The displacement operators satisfy the Weyl relations
\begin{align}
D(\xi) D(\eta) &=e^{-\frac{i}{2}\xi^T J\eta}D(\xi+\eta)\qquad\textrm{ for }\qquad \xi,\eta\in\mathbb{R}^{2N}\ , \label{eq:weylrelations}
\end{align}
and thus
\begin{align}
D(\xi) D(\eta)&=e^{-i\xi^T J\eta} D(\eta)D(\xi)\qquad\textrm{ for }\qquad \xi,\eta\in\mathbb{R}^{2N}\ .\label{eq:commutationdisplac}
  \end{align} 

An often-studied single-mode error model in linear optics is the Gaussian classical noise channel $\cN_{\sigma^2}:\cB(L^2(\mathbb{R}))\rightarrow\cB(L^2(\mathbb{R}))$ which randomly displaces a state in phase space according to a centered normal distribution with variance~$\sigma^2$. It is given by
\begin{align}
\cN_{\sigma^2}(\rho)&= \frac{1}{2\pi \sigma^2 }\cdot \int_{\mathbb{R}^{2}} e^{-\frac{\|\xi\|^2}{2\sigma^2}}D(\xi)\rho D(\xi)^\dagger d^{2}\xi\ .\label{eq:cNrho}
\end{align}
The channel~\eqref{eq:cNrho} is referred to as classical noise since we can think of it as first drawing a classical random variable~$Z\sim\cN(0,\sigma^2 I_2)$ according to a centered normal distribution, and then displacing by~$Z$.

In the following, we will assume that the physical-level noise on each mode is given by $\cN_{\sigma^2}$. In other words, for an $N$-mode system, the corresponding noise map is~$\cN_{\sigma^2}^{\otimes N}$.

  \subsection{Error recovery by modular measurements in the GKP state\label{subsec:singlegkp}}
  Gottesman-Kitaev-Preskill codes~\cite{gkp01} are designed to protect against displacement errors. To define oscillator-to-oscillator codes, we need a particular instance of this code family. In fact, we only need a certain bosonic stabilizer state   (spanning a $1$-dimensional subspace) instead of a stabilizer code.  For a single bosonic mode, the canonical GKP state~$\ket{\gkp}$ is the simultaneous $+1$-eigenstate of the two commuting displacement operators
  \begin{align}\label{eq:defstabilizersgkp}
S_Q\coloneqq e^{i\sqrt{2\pi}Q}\qquad\textrm{ and }\qquad S_P\coloneqq e^{i\sqrt{2\pi}P}\ .
\end{align}
  (The state~$\ket{\gkp}$ is uniquely defined up to a global phase which is irrelevant here.) We remark  that~$\ket{\gkp}$ is not an element of the Hilbert space~$L^2(\mathbb{R})$ but a distribution (see~\cite{bouzouina1996} for a rigorous treatment). In practice, $\ket{\gkp}$ needs to be replaced by an approximate GKP state with a finite average photon number. 
  
   It is instructive to consider how an unknown displacement error applied to the state~$\ket{\gkp}$ can be recovered from. Here we follow the same procedure as for a proper stabilizer code: A measurement of the stabilizer operators -- with outcome called the syndrome -- is followed by the (conditional) application of a unitary correction operation.

  \subsubsection{Syndrome measurement and errors} The fact that the operators~\eqref{eq:defstabilizersgkp} stabilize the state~$\ket{\gkp}$ means that $\ket{\gkp}$ has definite values of the position and momentum modulo~$\sqrt{2\pi}$. More precisely, measuring the stabilizer operators amounts to a joint measurement of the pair of ``modular'' operators
  \begin{align}
    (Q\mod \sqrt{2\pi},P\mod \sqrt{2\pi})\ ,
  \end{align}
  (which, unlike $Q$ and $P$ themselves, commute). This joint measurement can be realized by consuming two canonical GKP states and using Gaussian operations and quadrature measurements (cf.\ Fig.~\ref{fig:modulomeasurement}), see~\cite{gkp01}. When measuring the state~$\ket{\gkp}$, the measurement yields $(0,0)\in [-\sqrt{\pi/2},\sqrt{\pi/2})$.

Suppose an ideal GKP state is corrupted by an (unknown) displacement: We are given the state~$D(\xi)\ket{\gkp}$ for an unknown displacement (error) vector~$\xi\in\mathbb{R}^2$. Applying the modular measurement to the state~$D(\xi)\ket{\gkp}$ yields the pair $s=s(\xi)=(s_1,s_2)\in [-\sqrt{\pi/2},\sqrt{\pi/2})^2$ of syndromes  where
\begin{align}
  \begin{matrix}
s_1&=&\xi_1\mod \sqrt{2\pi}\ ,\\
s_2&=&\xi_2\mod \sqrt{2\pi}\ .
\end{matrix}\label{eq:syndromese}
\end{align}
Indeed, it follows  from the commutation relations~\eqref{eq:commutationdisplac} and the fact that the operators~\eqref{eq:defstabilizersgkp} stabilize~$\ket{\gkp}$ that $D(\xi)\ket{\gkp}$ is an eigenstate of $S_Q$ with eigenvalue~$e^{i\sqrt{2\pi}\xi_1}$ and of $S_P$ with eigenvalue~$e^{i\sqrt{2\pi}\xi_2}$. This implies~\eqref{eq:syndromese}.

\DeclareExpandableDocumentCommand{\phasenw}{O{}m}{|[phase,#1,label={[phase label,#1]#2}]| {} } 

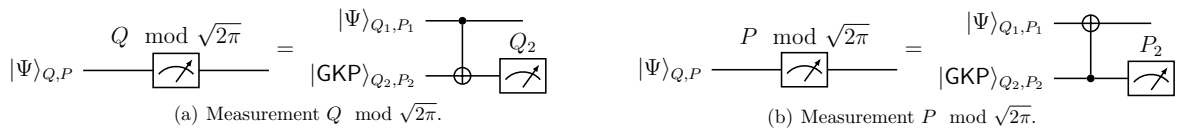
\begin{figure}[h]
\scalebox{0.73}{
\subfigure[Measurement $Q\mod \sqrt{2\pi}$.\label{fig:modulomeasurementq}]{\begin{quantikz}[row sep={1cm,between origins}]
\lstick{$\ket{\Psi}_{Q,P}$} &\meter{$Q \mod \sqrt{2\pi}$}&\qw
\end{quantikz}
=\begin{quantikz}[row sep={1cm,between origins}]
\lstick{$\ket{\Psi}_{Q_1,P_1}$} &\ctrl{1}&\qw \\
\lstick{$\ket{\gkp}_{Q_2,P_2}$} &\targ{} &\meter{$Q_2$}
\end{quantikz}\qquad\quad}\\ \\
\subfigure[Measurement $P\mod \sqrt{2\pi}$.\label{fig:modulomeasurementp}]{\begin{quantikz}[row sep={1cm,between origins}]
\lstick{$\ket{\Psi}_{Q,P}$} &\meter{$P \mod \sqrt{2\pi}$}&\qw
\end{quantikz}
=\begin{quantikz}[row sep={1cm,between origins}]
\lstick{$\ket{\Psi}_{Q_1,P_1}$} &\targ{}&\qw \\
\lstick{$\ket{\gkp}_{Q_2,P_2}$} &\ctrl{-1} &\meter{$P_2$}
\end{quantikz} }
}
\caption{Measurement of the quadrature operators modulo $\sqrt{2\pi}$. The left hand sides of the equalities in Figures~\ref{fig:modulomeasurementq} and~\ref{fig:modulomeasurementp} describe the measurement of $Q$ and $P$ modulo $\sqrt{2\pi}$ of a mode in the state $\ket{\Psi}\equiv\ket{\Psi}_{Q,P}$ respectively. This is realized by the circuits on the right hand side of the corresponding equality: In Figure~\ref{fig:modulomeasurementq}, the input state with associated quadratures $Q_1$, $P_1$ is coupled to an auxiliary mode in the state $\ket{\gkp}$ with associated quadratures $Q_2$, $P_2$ via the $\mathsf{SUM}_{1,2}$-gate ($=e^{-iQ_1P_2}$) depicted by the controlled-$\oplus$. This transforms the quadratures as $Q_1\mapsto Q_1$, $P_1\mapsto P_1-P_2$, $Q_2\mapsto Q_2+Q_1$, and $P_2\mapsto P_2$. As the second mode is initially in the state $\ket{\gkp}$ we have $Q_2+Q_1\mod \sqrt{2\pi}=Q_1\mod \sqrt{2\pi}$, and thus the output of the measurement of the $Q$-quadrature of the second mode modulo $\sqrt{2\pi}$ gives us the desired quantity. The measurement circuit for the $P$-quadrature follows the same strategy.}\label{fig:modulomeasurement}
\end{figure}

\subsubsection{Error recovery: maximum likelihood error decoding} When trying to use the syndrome information~\eqref{eq:syndromese} to figure out what the error~$\xi$ was (in order to then recover by applying the inverse displacement), one should take into account the prior distribution over errors. For example, if we are dealing with the classical noise channel~\eqref{eq:cNrho}, that is, if we are in fact measuring the noise-corrupted state~$\cN_{\sigma^2}(\proj{\gkp})$, then   the error~$\xi\in\mathbb{R}^2$ is distributed according to~$\cN(0,\sigma^2 I_2)$. The maximum likelihood error decoding problem then seeks to find the most likely error~$\hat{\xi}(s)$ consistent with the syndrome~$s=(s_1,s_2)$, i.e., satisfying $\hat{\xi}_1\mod\sqrt{2\pi}=s_1$ and $\hat{\xi}_2\mod\sqrt{2\pi}=s_2$. The corresponding recovery success probability~$\Pr_\xi[\hat{\xi}(s(\xi))=\xi]$ is well-studied: It is known as the informed unwrapping problem of modulo reduced Gaussian vectors (see Section~\ref{sec:informedunwrappingmodulo}). The recovery success probability directly gives the probability that the unitary correction operation~$D(\hat{\xi}(s(\xi))^{-1}$ applied after the error~$D(\xi)$ gives the identity and therefore the action of the latter is reversed.

The recovery procedure described here is analogous as the one used with GKP-codes encoding a finite-dimensional subspace. We note that in the special case considered here, where $1$-dimensional~$\mathbb{C}\ket{\gkp}$ 
is being protected, there is a simpler recovery strategy than the one described above: just inverting the displacement described by the syndrome vector does the job. This is due to the fact that we are considering a $1$-dimensional space: reverting the syndrome displacement results in a net displacement belonging to the stabilizer group and thus leaves the state~$\ket{\gkp}$ invariant. 

\subsection{Oscillator-to-oscillator codes}\label{subsec:oscillatortooscillatorcodes}
Here we briefly review the definition of oscillator-to-oscillator codes introduced by  Noh, Girvin and Jiang in~\cite{NohGirvinJiang20arX,NohGirvinJiang20}. An oscillator-to-oscillator code encodes $K$~bosonic modes into~$N$ bosonic modes by means of a Gaussian $N$-mode unitary~$U^{(N)}$ according to the map~\eqref{eq:GKPencodingmap}.
The code is the image of this map and is isomorphic to~$L^2(\mathbb{R})^{\otimes K}$. Any  code state $\ket{\bar{\Psi}}\coloneqq U^{(N)}(\ket{\Psi}\otimes\ket{\gkp}^{\otimes N-K})$ is clearly stabilized by the operators
\begin{align}
  U^{(N)}S_{Q_j}(U^{(N)})^\dagger\qquad\textrm{ and }\qquad  U^{(N)}S_{P_j}(U^{(N)})^\dagger\qquad\textrm{ where }\qquad j=K+1,\ldots,N\ , \label{eq:stabilizergeneratorsgkpstabilizer}
\end{align}
and where $S_{Q_j}\coloneqq e^{i\sqrt{2\pi}Q_j}$ and $S_{P_j}\coloneqq e^{i\sqrt{2\pi}P_j}$. We note that the term GKP-stabilizer code is sometimes used for these codes, although this is somewhat ambiguous since the original GKP code by Gottesman, Kitaev, and Preskill~\cite{gkp01} is itself a kind of stabilizer code. We will stick to the term oscillator-to-oscillator code here.

In the following, we  assume that $U^{(N)}=U_S$ is the Gaussian operation associated with a symplectic matrix $S\in Sp(2N,\mathbb{R})$, cf.~\eqref{eq:quadraticoperations}.

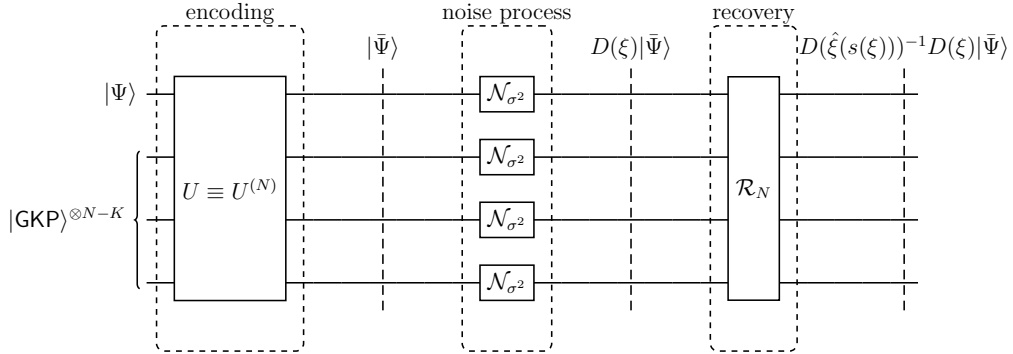
\begin{figure}[h]
\centering
\scalebox{0.73}{
  \begin{quantikz}
  \lstick{$\ket{\Psi}$}  & \gate[wires=4]{U\equiv U^{(N)}}\gategroup[4,steps=1,style={dashed,rounded corners,inner xsep=5pt,inner ysep=22pt},background]{{encoding}}& \qw & \qw &\qw \slice[style=black]{$\ket{\bar{\Psi}}$} & \qw & \qw & \qw &\gate{\cN_{\sigma^2}} \gategroup[4,steps=1,style={dashed,rounded corners,inner xsep=5pt,inner ysep=22pt},background]{{noise process}}& \qw & \qw & \qw \slice[style=black]{$D(\xi)\ket{\bar{\Psi}}$} &\qw & \qw & \qw &\gate[wires=4]{\cR_N}\gategroup[4,steps=1,style={dashed,rounded corners,inner xsep=5pt,inner ysep=22pt},background]{{recovery}}& \qw & \qw & \qw &\qw \slice[style=black]{$D(\hat{\xi}(s(\xi)))^{-1}D(\xi)\ket{\bar{\Psi}}$} &\qw  \\
  \lstick[3]{$\ket{\gkp}^{\otimes N-K}$}   & &\qw & \qw & \qw & \qw & \qw & \qw &\gate{\cN_{\sigma^2}}&\qw & \qw & \qw &\qw & \qw & \qw & &\qw & \qw & \qw & \qw &\qw \\
 & &\qw & \qw & \qw &\qw & \qw & \qw &\gate{\cN_{\sigma^2}} &\qw & \qw & \qw &\qw & \qw & \qw & & \qw & \qw & \qw & \qw &\qw\\
  & &\qw & \qw & \qw &\qw & \qw & \qw &\gate{\cN_{\sigma^2}} &\qw & \qw & \qw  &\qw & \qw & \qw & & \qw & \qw & \qw & \qw &\qw
   \end{quantikz}}
  \caption{Error correction circuit for oscillator-to-oscillator codes. The encoding consists of the application of the $N$-mode unitary $U^{(N)}$ to the $K$-mode input state $\ket{\Psi}$ and $N-K$ modes in the canonical GKP state $\ket{\gkp}$, yielding the encoded state $\ket{\bar{\Psi}}$. In the subsequent noise process, the Gaussian classical noise channel $\cN_{\sigma^2}$, which applies a random displacement according to a centered normal distribution with variance $\sigma^2$, is applied to each mode. This process can equivalently be described as  drawing a $2N$-dimensional real vector $\xi\sim\cN(0,\sigma^2 I_{2N})$  according to the  centered normal distribution with variance~$\sigma^2$, and  subsequently applying the displacement~$D(\xi)$ to the encoded state~$\ket{\bar{\Psi}}$. The last step, i.e., the recovery, is described in more detail in Fig.~\ref{fig:decoding}.}\label{fig:eccircuit}
\end{figure}

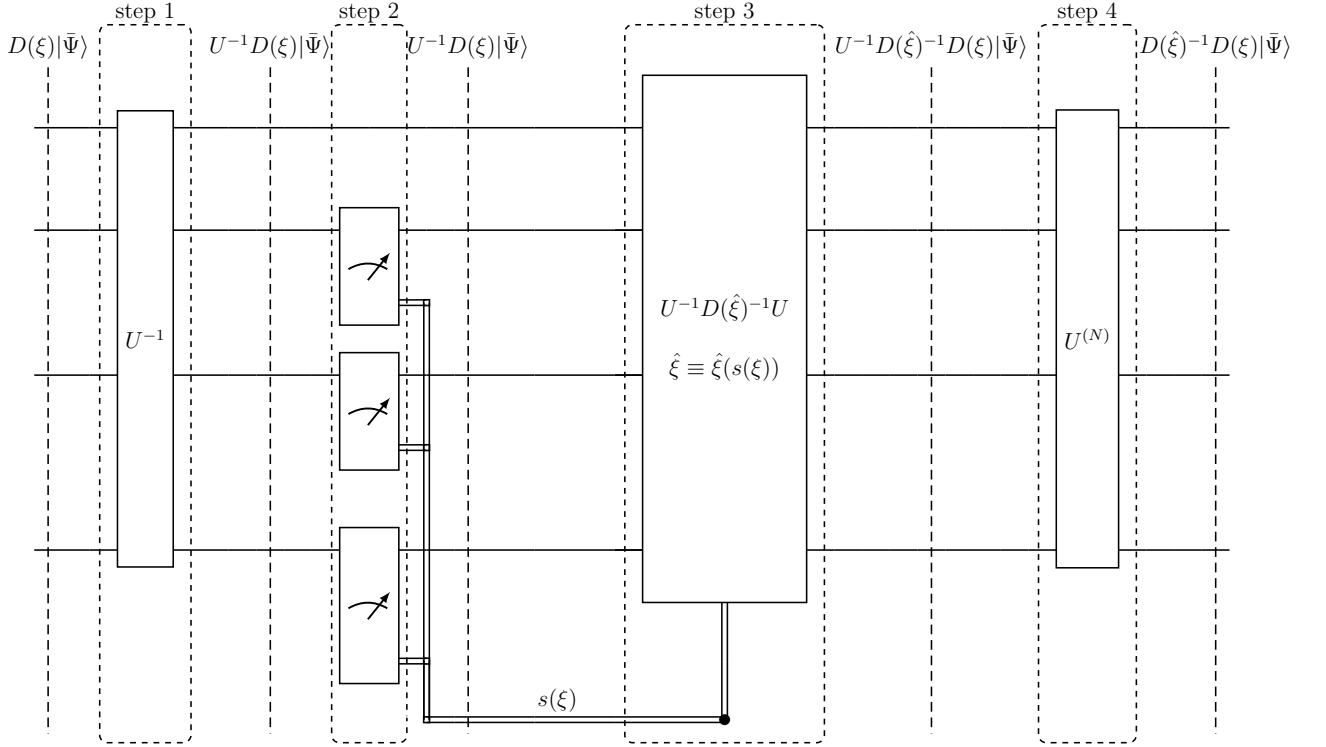
\begin{figure}[h]
\centering
\scalebox{0.73}{
\begin{quantikz}
\slice[style=black]{$D(\xi)\ket{\bar{\Psi}}$} 
&\qw &\qw & \gate[6,nwires={3,5}]{U^{-1}}\gategroup[7,steps=1,style={dashed,rounded corners,inner xsep=5pt,inner ysep=22pt},background]{{step $1$}}
&\qw &\qw &\slice[style=black]{$U^{-1}D(\xi)\ket{\bar{\Psi}}$}\qw&\qw &\qw  
&\qw \gategroup[7,steps=1,style={dashed,rounded corners,inner xsep=0pt,inner ysep=22pt},background]{{step $2$}} 
& \qw &\qw \slice[style=black]{$U^{-1}D(\xi)\ket{\bar{\Psi}}$}&\qw &\qw 
&\gate[6,nwires={3,5}]{\begin{array}{c} U^{-1} D(\hat{\xi})^{-1} U \\ \\ \hat{\xi}\equiv \hat{\xi}(s(\xi)) \end{array}}\gategroup[7,steps=1,style={dashed,rounded corners,inner xsep=5pt,inner ysep=22pt},background]{{step $3$}} 
& \qw &\qw &\qw & \qw \slice[style=black]{$U^{-1}D(\hat{\xi})^{-1}D(\xi)\ket{\bar{\Psi}}$} &\qw &\qw &\qw &\qw 
&\gate[6,nwires={3,5}]{U^{(N)}}\gategroup[7,steps=1,style={dashed,rounded corners,inner xsep=5pt,inner ysep=22pt},background]{{step $4$}}
& \qw &\qw &\qw \slice[style=black]{$D(\hat{\xi})^{-1}D(\xi)\ket{\bar{\Psi}}$}   &\qw \\
&\qw &\qw & &\qw &\qw &\qw &\qw &\qw  &\gate[2,nwires={2}]{\begin{tikzpicture}
\draw[->,line width=1.0pt,-latex] (0.5,0) -- (0.9, 0.5);
\coordinate (G) at (0.15,0.2);
\coordinate (R) at (0.85,0.2);
\path [bend right,line, line width=1.0pt] (R) edge (G);
\end{tikzpicture}} &\qw & \qw &\qw &\qw &\qw & 			 	\qw &\qw  & \qw &\qw &\qw &\qw &\qw &\qw & & \qw & \qw &\qw &\qw\\
& & & & & & & &  & &\cwbendmod{2} &  & & & & 			 	 &  &  & & & & & & &  &  & &\\
&\qw &\qw & &\qw &\qw &\qw &\qw &\qw  &\gate[2,nwires={2}]{\begin{tikzpicture}
\draw[->,line width=1.0pt,-latex] (0.5,0) -- (0.9, 0.5);
\coordinate (G) at (0.15,0.2);
\coordinate (R) at (0.85,0.2);
\path [bend right,line, line width=1.0pt] (R) edge (G);
\end{tikzpicture}}  &\qw  & \qw &\qw &\qw &\qw & 			 	\qw &\qw  & \qw &\qw &\qw &\qw &\qw &\qw & & \qw & \qw &\qw &\qw\\
& & & & & & & &  & &\cwbendmod{2} &  & & & & 			 	 &  &  & & & & & & &  &  & &\\
&\qw &\qw & &\qw &\qw &\qw &\qw &\qw  &\gate[2,nwires={2}]{\begin{tikzpicture}
\draw[->,line width=1.0pt,-latex] (0.5,0) -- (0.9, 0.5);
\coordinate (G) at (0.15,0.2);
\coordinate (R) at (0.85,0.2);
\path [bend right,line, line width=1.0pt] (R) edge (G);
\end{tikzpicture}}  &\qw & \qw &\qw &\qw &\qw & 			 	\qw &\qw  & \qw &\qw &\qw &\qw &\qw &\qw & & \qw & \qw &\qw &\qw\\
& & & & & & & &  & &\cwbendmod{1} &  &\begin{array}{c} \\ \\ \qquad \vspace{-0.4cm}s(\xi) \end{array}   & & \vcw{-1} &			 	 &  &  & & & & & & &  &  & &\\
&    &    & &    &    &    &    &    &	 &  \cwbendmodmod{-1}&\cw &\cw & \cw&\cwbend{-1} &    &    &     &    &    &    &    &    & &     &     &    &
   \end{quantikz}}
  \caption{Recovery step of the error correction circuit for oscillator-to-oscillator codes (cf.\ Fig.~\ref{fig:eccircuit}). The input is the corrupted state $D(\xi)\ket{\bar{\Psi}}$ after the noise process. In step $1$ the encoding is reversed, resulting in the state $(U^{(N)})^{-1}D(\xi)\ket{\bar{\Psi}}$. In step $2$, the syndrome $s(\xi)$ is measured, i.e., for every mode $j=K+1,\dots,N$ the quadrature operators $Q_j=R_{2j-1}$, $P_j=R_{2j}$ are measured modulo $\sqrt{2\pi}$ (cf.\ Fig.~\ref{fig:modulomeasurement}) yielding the entries $s_{2j-1-2K}$, $s_{2j-2K}$ of the $2(N-K)$-dimensional vector~$s$ respectively.
   The measurement result is a deterministic function~$s=s(\xi)$ of the error~$\xi$ (see Eq.~\eqref{eq:sjdetermin}). In particular, the measurement does not change the state. The classical syndrome measurement outcome is then used to compute the correction operation $(U^{(N)})^{-1}D(\hat{\xi}(s(\xi)))^{-1}U^{(N)}$, which is applied in step $3$ and yields the state $(U^{(N)})^{-1}D(\hat{\xi}(s(\xi)))^{-1}D(\xi)\ket{\bar{\Psi}}$. Finally, in step $4$, the corrected state is encoded again. The output is the state $D(\hat{\xi}(s(\xi)))^{-1}D(\xi)\ket{\overline{\Psi}}$.}\label{fig:decoding}
\end{figure}

\subsubsection{Syndrome measurement and errors}
A (simultaneous) measurement of the family of commuting operators~\eqref{eq:stabilizergeneratorsgkpstabilizer}
is equivalent to the joint measurement of the commuting set of the $2(N-K)$ modular operators
\begin{align}
\left(\sum_{k=1}^{2N} S_{j,k} R_k\right)\mod \sqrt{2\pi}\qquad\qquad\textrm{ where }\qquad j=2K+1,\ldots,2N\ ,\label{eq:modularoperatorstobemeasured}
\end{align}
according to Eq.~\eqref{eq:quadraticoperations}. Measuring a corrupted code state~$D(\xi)\ket{\bar{\Psi}}$, where $\xi\in\mathbb{R}^{2N}$, yields the syndrome $s=(s_1,\ldots,s_{2(N-K)})\in [-\sqrt{\pi/2},\sqrt{\pi/2})^{2(N-K)}$ with 
\begin{align}
s_j &= \left(\sum_{k=1}^{2N} S_{2K+j,k} \xi_k\right)\mod \sqrt{2\pi}\qquad\qquad\textrm{ for }\qquad j=1,\ldots,2(N-K)\ .\label{eq:sjdetermin}
\end{align}
(This again follows from the commutation relations~\eqref{eq:commutationdisplac}.) We will write $s=s(\xi)$ to emphasize that the syndrome is a function of the displacement vector~$\xi$. 
This joint measurement  can for example be implemented by applying $(U^{(N)})^{-1}$ to the state to be measured
using the measurement circuits mentioned in Section~\ref{subsec:singlegkp} to
measure the one-mode modular operators $(S_{Q_j},S_{P_j})$ for $j=K+1,\ldots,N$, and applying~$U^{(N)}$ again to the post-measurement state (cf.\ Fig.~\ref{fig:decoding}). 

\subsubsection{Error recovery and logical error}
Recovery from  a displacement error~$D(\xi)$ proceeds as in any stabilizer code by first extracting the syndrome~$s\in [-\sqrt{\pi/2},\sqrt{\pi/2})^{2(N-K)}$ by measurement, and subsequently applying a unitary correction operation to the post-measurement state. In the case under consideration, a recovery strategy is specified by an estimator (function)
\begin{align}
  \begin{matrix}
    \hat{\xi}\colon &[-\sqrt{\pi/2},\sqrt{\pi/2})^{2(N-K)}&\rightarrow &\mathbb{R}^{2N}\\
    & s & \mapsto & \hat{\xi}(s)
    \end{matrix}\label{eq:estimatorfct}
  \end{align}
for the actual error~$\xi$, based on the syndrome~$s$. When observing the syndrome~$s$,  the correction operation~$D(\hat{\xi}(s))^{-1}$ is applied. The function~$\hat{\xi}$  is typically chosen such that the syndrome for the error $D(\hat{\xi}(s))$ is identical to~$s$ (i.e., $s(\hat{\xi}(s))=s$), to ensure that this operation returns the state to the code space.

The resulting state when starting from the corrupted state $D(\xi)\ket{\bar{\Psi}}$ is thus 
\begin{align}
D(\hat{\xi}(s(\xi)))^{-1}D(\xi)\ket{\bar{\Psi}}\ .
\end{align}
Up to an irrelevant global phase (cf.~\eqref{eq:weylrelations}), the resulting effect on the physical modes is  a displacement by the vector~
\begin{align}
  \eta=\eta(\xi)\coloneqq -\hat{\xi}(s(\xi))+\xi\in\mathbb{R}^{2N}\ .\label{eq:etaxidef}
  \end{align}

\paragraph{Maximum likelihood error decoding.} We will assume that  we are dealing with (independent and identical) classical noise with variance~$\sigma^2$ on each mode, see Eq.~\eqref{eq:cNrho}. That is, the noise-corrupted encoded state is $\cN_{\sigma^2}^{\otimes N}(\proj{\bar{\Psi}})$. In this case, the displacement error vector~$\xi\in\mathbb{R}^{2N}$ has a centered normal prior distribution, i.e.,~$\xi\sim \cN(0,\sigma^2 I_{2N})$. Maximum likelihood error decoding then amounts to choosing the error that is most likely given the observed syndrome~$s\in [-\sqrt{\pi/2},\sqrt{\pi/2})^{2(N-K)}$. In other words, the estimator function is given by
\begin{align}
\hat{\xi}^{\mathsf{ML}}(s)&\coloneqq \argmax_{\xi\in\mathbb{R}^{2N}} f_{Z|s(Z)=s}(\xi)\ ,\label{eq:argmaxexpr}
\end{align}
where $f_{Z|s(Z)=s}$ is the conditional probability density function when~$Z\sim \cN(0,\sigma^2 I_{2N})$ is conditioned on~$s(Z)=s$.  
Here and below, ties in expressions such as~\eqref{eq:argmaxexpr} when taking $\argmax$ are broken arbitrarily. By Bayes' rule, this definition is equivalent to
\begin{align}
\hat{\xi}^{\mathsf{ML}}(s)&=\argmax_{\xi\in \{z\in\mathbb{R}^{2N}\ |\ s(z)=s\}} f_Z(\xi)\ .\label{eq:maximumlikelihoodximl}
  \end{align}

\paragraph{Logical error and figure of merit.} Recall that the overall effect of a displacement error and subsequent correction is given by a displacement vector~$\eta$ as in~\eqref{eq:etaxidef}. To see what the effect on the logical information is, let us assume that $U^{(N)}=U_S$ is given by the symplectic matrix~$S$. Observe that for $x=S\eta$ we have
\begin{align}
  D(\eta)\ket{\overline{\Psi}}&= U_S D(x)\left(\ket{\Psi}\otimes \ket{\gkp}^{\otimes N-K}\right)\\
  &=U_S\left(D(x_A)\ket{\Psi}\otimes D(x_B)\ket{\gkp}^{\otimes N-K}\right)\ ,
\end{align}
where we used the fact that $D(\eta)U_S=U_SD(S\eta)$, and where we write
$x=(x_A,x_B)\in\mathbb{R}^{2K}\times\mathbb{R}^{2(N-K)}\cong\mathbb{R}^{2N}$.
Under the assumption that $D(\hat{\xi}(s))$ causes syndrome~$s$, the overall displacement~$D(\eta)$ maps the code space to itself, and thus $D(x_B)$ stabilizes~$\ket{\gkp}^{\otimes N-K}$. In particular, we conclude that
\begin{align}
  D(\eta)U_S\left(\ket{\Psi}\otimes\ket{\gkp}^{\otimes N-K}\right)&= U_S\left(D(x_A)\ket{\Psi}\otimes \ket{\gkp}^{\otimes N-K}\right)\ ,
  \end{align}
which  shows that the logical error is a displacement by~$x_A$, i.e., by the vector~$x_A\in\mathbb{R}^{2K}$ consisting of the first $2K$~components of~$S(-\hat{\xi}(s(\xi))+\xi)$.

Note that $x_A=x_A(\xi)$ is a deterministic function of the displacement error vector~$\xi$ depending on the choice of the estimator function~\eqref{eq:estimatorfct}. In the case where~$\xi\sim\cN(0,\sigma^2 I_{2N})$, $x_A$ is a random variable supported on~$\mathbb{R}^{2K}$. To quantify error suppression, we may use e.g., the variance of this random variable (as in~\cite{NohGirvinJiang20arX,NohGirvinJiang20} for example). Here we use a slightly different measure: for any~$\epsilon>0$, we set
\begin{align}
P_{\textrm{succ}}(\epsilon)&\coloneqq \Pr_{\xi\sim\cN(0,\sigma^2 I_{2N})}\left[\|x_A(\xi)\|\leq \epsilon\right]\ ,\label{eq:psuccepsdef}
  \end{align}
where $\|y\|\coloneqq \left(\sum_{j=1}^{2K}|y_j|^2\right)^{1/2}$ is the Euclidean norm of $y\in\mathbb{R}^{2K}$. In other words, we are interested in the probability that the resulting logical error after recovery belongs to an $\epsilon$-ball in~$\mathbb{R}^{2K}$ centered at the origin. We call $P_{\textrm{succ}}(\epsilon)$ the decoding success probability. 

\section{Analysis of maximum likelihood error recovery for\\ oscillator-to-oscillator codes}\label{sec:problemreformulation}
In this section, we reformulate the recovery success probability~\eqref{eq:psuccepsdef} with respect to maximum likelihood error decoding of an oscillator-to-oscillator code subject to Gaussian classical noise in terms of a purely classical estimation problem in Section~\ref{sec:reformulationrec}. 
 In Section~\ref{sec:uppbnd}, we then summarize our main upper bound on this quantity. 

\subsection{Reformulating the recovery success probability\label{sec:reformulationrec}}
Our goal is to establish an upper bound on the quantity~\eqref{eq:psuccepsdef}
when maximum likelihood estimation of the (physical) error is used. To this end, we first reformulate this quantity in terms of an estimation problem of purely classical information-theoretic nature.

This reformulation is guided by the measurement procedure for
the operators~\eqref{eq:modularoperatorstobemeasured}.  Suppose the code state~$\ket{\overline{\Psi}}$ undergoes a displacement~$D(\xi)$. Using that  $U_S^\dagger D(\xi)=D(S\xi)U_S^\dagger$, we have (similarly as before)
\begin{align}
U_S^\dagger D(\xi)\ket{\overline{\Psi}}&=D(x)(\ket{\Psi}\otimes\ket{\gkp}^{\otimes N-K})\qquad\textrm{ where }\qquad x=S\xi\in\mathbb{R}^{2N}\ .
\end{align}
Writing $x=(x_A,x_B)\in\mathbb{R}^{2K}\times\mathbb{R}^{2(N-K)}$, we have  $D(x)=D(x_A)\otimes D(x_B)$, and thus measurement of GKP-state stabilizers $S_{Q_j}, S_{P_j}$ on the modes~$j\in \{K+1,\ldots,N\}$ provides the syndrome
\begin{align}
s&= x_B\mod \sqrt{2\pi}\in [-\sqrt{\pi/2},\sqrt{\pi/2})^{2(N-K)}\ ,
\end{align}
where the modulo-operation is applied to each entry of $x_B$.

If the original error is distributed according to~$\xi\sim\cN(0,\sigma^2I_{2N})$, then $x\sim \cN(0,\sigma^2 (S^TS)^{-1})$ by definition of~$x$.  Furthermore,
if $\hat{\xi}^{\mathsf{ML}}$ is the maximum likelihood estimator for~$\xi$ (based on the syndrome~$s$ (cf.~\eqref{eq:maximumlikelihoodximl})), then $\hat{x}^{\mathsf{ML}}=S\hat{\xi}^{\mathsf{ML}}$ is the maximum likelihood estimator for~$x$. It can be written as
\begin{align}
  \hat{x}^{\mathsf{ML}}(s)&=\argmax_{x\in \{x\in\mathbb{R}^{2N}\ |\ x_B\mod \sqrt{2\pi}=s\}} f_X(x)\ ,
\end{align}
where $f_X$ is the distribution function of the random variable~$X\sim \cN(0,\sigma^2 (S^TS)^{-1})$. In particular, the residual logical error after error correction is given by the first $2K$~components of the vector~$-\hat{x}^{\mathsf{ML}}(s)+x$. We have thus arrived at the following reformulation:
Let $\Pi_A:\mathbb{R}^{2N}\rightarrow\mathbb{R}^{2K}$ be the projection map taking a vector~$x\in\mathbb{R}^{2N}$ to a vector with only its first $2K$~components.
Similarly, let $\Pi_B:\mathbb{R}^{2N}\rightarrow\mathbb{R}^{2(N-K)}$ denote the projection map onto the last $2(N-K)$~components. Then the decoding success probability~\eqref{eq:psuccepsdef} (when using maximum likelihood error decoding) can be written as 
\begin{align}
  P_{\textrm{succ}}(\epsilon)&=\Pr_{X\sim\cN(0,\sigma^2 (S^TS)^{-1})}\left[\left\|
\Pi_A(X)-\Pi_A\left(\hat{x}^{\mathsf{ML}}(\Pi_B X \mod \sqrt{2\pi})\right)\right\|\leq \epsilon\right]\ .
  \end{align}
In other words, $P_{\textrm{succ}}(\epsilon)$ can be understood
as the probability of $\epsilon$-approximately estimating the first $2K$~components of a random Gaussian vector~$X$, given the remaining $2(N-K)$~modulo-$\sqrt{2\pi}$-reduced components.

\subsection{An upper bound on the decoding success probability\label{sec:uppbnd}}
In Section~\ref{sec:classicalanalysis} we will analyze this estimation problem more generally. We  consider a random vector~$X\in\cN(0,\Sigma)$ on $\mathbb{R}^n$ drawn according to a $n$-variate centered normal distribution with covariance matrix $\Sigma$, and study the probability of $\epsilon$-correctly estimating the first $k$~components of~$X$, given the modulo-$\Delta$-reduced values of the remaining~$n-k$ components (using maximum likelihood estimation of the entire vector). One of our main results (Corollary~\ref{cor:upperboundpsepsilon}) concerning this problem is the upper bound
\begin{align}
P_{\textrm{succ}}(\epsilon)&\leq \mu_Z\left(\cB_{\epsilon/\sqrt{\lambda_{\min}(\Sigma)}}(0)\right)\ ,
\end{align}
where $\mu_Z$ is the probability measure of a centered normal Gaussian vector~$Z\sim\cN(0,I_k)$, and where $\cB_\delta(0)\subset\mathbb{R}^k$ is the closed $\delta$-ball with respect to the Euclidean norm.
For an oscillator-to-oscillator code encoding $K$ logical modes into $N$~physical modes, we have $(n,k)=(2N,2K)$ and \begin{align}\label{eq:specificationDS}
(\Delta,\Sigma)=(\sqrt{2\pi},\sigma^2(S^TS)^{-1})\ .
\end{align}
Since $1/\lambda_{\min}(\Sigma)=\lambda_{\max}(\Sigma^{-1})$ we obtain (cf.~\eqref{eq:squeezingmeasureU}) the following:
\begin{theorem}
  Let~$\epsilon>0$. The decoding success probability~$P_{\textrm{succ}}(\epsilon)$ when encoding $K$ into~$N$ modes using an oscillator-to-oscillator code with a Gaussian encoding unitary~$U$ satisfies
  \begin{align}
P_{\textrm{succ}}(\epsilon)&\leq \mu_Z\left(\cB_{\epsilon\cdot \mathsf{sq}(U)/\sigma}(0)\right)\ ,
  \end{align}
  where $\mu_Z$ is the probability measure of a centered normal Gaussian vector~$Z\sim\cN(0,I_{2K})$.
  \end{theorem}
  
Since the norm of a $2K$-dimensional centered normal vector is close to $\sqrt{2K}$ with high probability (see e.g.,~\cite[Theorem~3.1.1]{vershynin2018}) this implies that the amount of squeezing~$\mathsf{sq}(U)$ has to grow roughly as~$\sqrt{2K}\sigma/\epsilon$ in the decoding error~$\epsilon$ in order to achieve a success probability close to~$1$. This establishes our no-go result: with a bounded amount of squeezing,  an arbitrarily small decoding error cannot be achieved.

\section{Unwrapping Modulo Reduced Gaussian Vectors}\label{sec:classicalanalysis}
In this section, we derive our main result concerning a classical estimation problem. We refer to it as partial unwrapping of modulo reduced Gaussian vectors.

In Section~\ref{sec:informedunwrappingmodulo}, we first review the informed unwrapping problem of modulo reduced Gaussian vectors, which has been studied before. In Section~\ref{sec:partialinformedunwrapping}, we then introduce the partial unwrapping problem. In Section~\ref{sec:upperboundsoverview}, we give a high-level overview of how we obtain an upper bound on the corresponding figure of merit, the decoding success probability. 

We then derive our technical results: 
In Section~\ref{sec:decsuccvoronoi}, we show how the decoding success probability can be written as the probability mass of a countable union of displaced and thickened degenerate Voronoi cells under a centered unit-variance Gaussian normal distribution. Here the relevant displacements are given by a lattice.
 In Section~\ref{subsec:cp:propdegvoronoi}, we describe a polytope that is an outer bound on the degenerate Voronoi cell.
A parametrization of this polytope
(exploiting the fact that it is lower-dimensional) is described in Section~\ref{sec:parametrizationPLambda}. 
In Section~\ref{sec:translationparametrization}, we show how to parametrize the union of translates of this polytope (up to zero measure sets). In Section~\ref{sec:displacedthick}, we show how to integrate over thickened versions of these polytopes, and give a general upper bound on the probability mass of the union of thickened and displaced polytopes. Finally, in Section~\ref{sec:upperboundunwrapp}, we combine these results to obtain a general upper bound on the decoding success probability for the partial unwrapping problem.

\subsection{Informed unwrapping\label{sec:informedunwrappingmodulo}} 
Let $n>0$ be an integer, $\Delta>0$ a real number and $\Sigma\in\mathsf{Mat}_{n\times n}(\mathbb{R})$ the covariance matrix of a normal distribution on~$\mathbb{R}^n$. We note that such matrices are positive definite and symmetric, a fact we will use below. 
Let $X=(X_1,\ldots,X_n)\sim\cN(0,\Sigma)$ be a random $n$-dimensional vector drawn according to the centered normal distribution with covariance matrix~$\Sigma$. Let $X^*$ be the $n$-dimensional vector in~$[-\Delta/2,\Delta/2)^{n}$ whose $j$-th entry $X_j^*$ is obtained by reducing the coordinate $X_j$ modulo~$\Delta$.

The {\em informed unwrapping problem of modulo reduced Gaussian vectors} asks to reconstruct~$X$ from~$X^*$. This problem arises naturally in signal processing since the modulo operation applied to $x\in\mathbb{R}$ amounts to discarding the most significant bits in the binary representation of~$x$. The task is thus to reconstruct the Gaussian vector from the least significant bits of its coordinates.  The expression {\em informed} refers to the fact that the covariance matrix~$\Sigma$ is known to the decoder. 

The MAP (for maximum a posteriori likelihood) estimator $\hat{x}^{\mathsf{MAP}}:[-\Delta/2,\Delta/2)^{n}\rightarrow\mathbb{R}^n$ for this problem is the function that maximizes the decoding probability~$\Pr\left[\hat{x}(X^*)=X\right]$, i.e., 
\begin{align}
\hat{x}^{\mathsf{MAP}}(z)&=\argmax_{y\in\mathbb{R}^n:y^*=z}f_\Sigma(y)\ ,\quad z\in [-\Delta/2,\Delta/2)^{n}\ , \label{eq:mapestimator}
\end{align}
where 
\begin{align}\label{eq:probdensfctsigma}
f_\Sigma(x)=\frac{1}{(2\pi)^{n/2}(\det \Sigma)^{1/2}} e^{-\frac{1}{2}x^T\Sigma^{-1}x}
\end{align}
is the probability density function of $X$.  It can be shown (see~\cite[Section III]{romanov2019blind}) that the corresponding success probability is related to the probability that the random variable $Z=\Sigma^{-1/2}X\sim \cN(0,I_n)$ belongs to the Voronoi region 
\begin{align}
\cV:=\{z\in\mathbb{R}^n\ |\ \|z\|\leq \|z-\Delta\Sigma^{-1/2}  b\|\textrm{ for all }b\in\mathbb{Z}^n\}\ 
\end{align}
of the lattice~$\Delta\Sigma^{-1/2}\mathbb{Z}^n$, i.e.,
\begin{align}
\Pr\left[\hat{x}^{\mathsf{MAP}}(X^*)=X\right]&=\Pr\left[Z\in\cV\right]\ .\label{eq:mapdecodingsuccess}
\end{align}
Here $\|x\|:=\sqrt{\sum_{j=1}^n x_j^2}$ denotes the Euclidean norm of $x\in\mathbb{R}^n$. With~\eqref{eq:mapdecodingsuccess}, the decoding error probability can be related to the quantity~$\Delta/(\det \Sigma)^{n/2}$, see~\cite{ordentlicherezinteger,outageprob,modulobasedarchitecture}.

\subsection{Partial informed unwrapping: statement of the problem\label{sec:partialinformedunwrapping}}

Here we consider a variant of the above problem which we refer to as partial informed unwrapping of modulo reduced Gaussian vectors. In this variant, one is given $n-k$~modulo reduced entries of a random vector $X=(X_1,\ldots,X_n)\sim\cN(0,\Sigma)$ (where $0<k<n$). The task is to reconstruct the $k$~remaining coordinates of~$X$.

To describe this in more detail, let us write~$n$-dimensional vectors as $x=(x_1,\ldots,x_n)=\binom{x_A}{x_B}\in\mathbb{R}^k\times\mathbb{R}^{n-k}$ where $x_A=(x_1,\ldots,x_k)$ and $x_B=(x_{k+1},\ldots,x_{n})$. It will also be useful to introduce corresponding projection maps
\begin{align}
\begin{matrix}
\Pi_A:&\mathbb{R}^n&\rightarrow &\mathbb{R}^k\\
& \binom{x_A}{x_B}&\mapsto & x_A
\end{matrix}\qquad\textrm{ and }\qquad \begin{matrix}
\Pi_B:&\mathbb{R}^n&\rightarrow &\mathbb{R}^{n-k}\\
& \binom{x_A}{x_B}&\mapsto & x_B\ .
\end{matrix}
\end{align}
For a realization $x=\binom{x_A}{x_B}\in\mathbb{R}^n$ of $X$, let 
\begin{align}\label{eq:xBstardefinition}
x_B^*\coloneqq \left((\Pi_B x)\mod \Delta\right)\in [-\Delta/2,\Delta/2)^{n-k}\ ,
\end{align}
where the modulo operation is applied entrywise.

Let $X=\binom{X_A}{X_B}\sim\cN(0,\Sigma)$. The partial informed unwrapping problem
asks to reconstruct~$X_A$ from $X^*_B$. For a given estimator $h:[-\Delta/2,\Delta/2)^{n-k}\rightarrow \mathbb{R}^k$ and an error tolerance~$\epsilon>0$, we call 
\begin{align}
p^h_{S}(\epsilon)&=\Pr\left[\|h(X^*_B)-X_A\|\leq\epsilon \right]
\end{align}
the decoding success probability of~$h$. A natural choice of decoder for this problem relies on the estimator (analogous to Eq.~\eqref{eq:mapestimator}) 
 $\hat{x}^{\mathsf{MAP}}:[-\Delta/2,\Delta/2)^{n-k}\rightarrow\mathbb{R}^n$ 
 for $x$ defined by 
\begin{align}
\hat{x}^{\mathsf{MAP}}(z)&=\argmax_{y\in\mathbb{R}^n:y_B^*=z}f_\Sigma(y)\label{eq:mapestimatortwo}
\end{align}
and uses the concatenation~$h^{\mathsf{MAP}}:= \Pi_A\circ \hat{x}^{\mathsf{MAP}}$ of this map with the projection map~$\Pi_A$. We call this the MAP-decoder. In the following, we will study the associated decoding success probability denoted by~$p_S(\epsilon)$.

\subsection{Upper bounds on partial informed unwrapping\label{sec:upperboundsoverview}} 
Here we give a high-level overview of our derivation of an upper bound for the decoding success probability $p_S(\epsilon)$ in Sections~\ref{sec:decsuccvoronoi}-\ref{sec:upperboundunwrapp}. 

To find the upper bound,
we first establish a connection to the probability mass (under the centered standard normal distribution) of certain subsets of~$\mathbb{R}^n$ analogous to Eq.~\eqref{eq:mapdecodingsuccess}. To state this, let
\begin{align}
\cV_{\Sigma^{-1/2}}:=\left\{z\in\mathbb{R}^n\ \left|\ \|z\|\leq \left\|z-\Sigma^{-1/2} \binom{h}{\Delta b}\right\| \textrm{ for all }(h,b)\in\mathbb{R}^k\times\mathbb{Z}^{n-k}\right.\right\}\  \label{eq:voronoicelldef}
\end{align}
denote the Voronoi cell of the (degenerate) lattice $\Sigma^{-1/2}(\mathbb{R}^k\times \Delta\mathbb{Z}^{n-k})$.
Let
\begin{align}
\cV_{\Sigma^{-1/2}}(\epsilon)\coloneqq \cV_{\Sigma^{-1/2}}+\Sigma^{-1/2} \left(\cB_\epsilon(0)\times\{0\}^{n-k}\right)\ .\label{eq:epsilonthickening}
\end{align}
be an $\epsilon$-``thickening'' of $\cV_{\Sigma^{-1/2}}$. In this expression, $\cX+\cY:=\{x+y\ |\ x\in\cX,y\in\cY\}$ is the Minkowski sum of two subsets $\cX,\cY\subset\mathbb{R}^n$, and 
\begin{align}
\cB_\epsilon(y):=\{h\in\mathbb{R}^k\ |\ \|h-y\|\leq \epsilon\}\
\end{align}
is the closed $\epsilon$-ball around $y\in\mathbb{R}^k$. The relevant subset 
is a union of translates of the set~$\cV_{\Sigma^{-1/2}}(\epsilon)$ of the form
\begin{align}
\Sigma^{-1/2}\binom{0}{\Delta b}+\cV_{\Sigma^{-1/2}}(\epsilon)\qquad\textrm{ with }\qquad b\in\mathbb{Z}^{n-k}\ .
\end{align}
 That is, we have the following: The decoding success probability~$p_S(\epsilon)$ is given by
\begin{align}
p_S(\epsilon)&=\Pr\left[Z\in \Sigma^{-1/2}(\{0\}^k\times \Delta\mathbb{Z}^{n-k})+\cV_{\Sigma^{-1/2}}(\epsilon)\right]\qquad\textrm{ where }\qquad Z\sim\cN(0,I_n)\ . \label{eq:psepsilonexact}
\end{align}
We will prove~\eqref{eq:psepsilonexact} below, see Theorem~\ref{thm:mainvoronoireformulation}.

To proceed, we find an outer bound on the set~$\Sigma^{-1/2}(\{0\}^k\times \Delta\mathbb{Z}^{n-k})+\cV_{\Sigma^{-1/2}}(\epsilon)$. This allows us to obtain an upper bound on the quantity~\eqref{eq:psepsilonexact} which depends on the marginal distribution of $X_B=\Pi_B X$, and the conditional distribution of $X_A=\Pi_A X$ given that $X_B=x_B$. It is of the form
\begin{align}
p_S(\epsilon)&\leq \sum_{b\in\mathbb{Z}^{n-k}} \int_{\cR_b} f_{X_B}(\xi)\cdot\Pr\left[Y\in \cB_\epsilon(c(\xi,\Delta b))\right] d\xi\ ,\label{eq:psepsilonupperboundgeneral}
\end{align}
where $\{\cR_b\}_{b\in\mathbb{Z}^{n-k}}$ is a collection of measurable subsets of~$\mathbb{R}^{n-k}$
 such that $\sum_{b\in\mathbb{Z}^{n-k}} \int_{\cR_b} f_{X_B}(\xi)d\xi\leq 1$, the ball centers~$c(\xi,\Delta b)\in\mathbb{R}^{k}$ are located at
\begin{align}
c(\xi,\Delta b)&\coloneqq \Gamma_1\xi+\Gamma_2 \Delta b\ 
\end{align}
for certain matrices $\Gamma_1,\Gamma_2$ depending on~$\Sigma$ (see Theorem~\ref{thm:mainintegral}), and $Y\sim\cN(0,\Sigma_*)$ is a normal distributed random variable on~$\mathbb{R}^k$ whose covariance matrix is the Schur complement
\begin{align}
\Sigma_*&=\Sigma_{AA}-\Sigma_{AB}(\Sigma_{BB})^{-1}\Sigma_{BA}\ .
\end{align}
Expression~\eqref{eq:psepsilonupperboundgeneral}  depends on both~$\Sigma$ and~$\Delta$ in general. With Anderson's inequality~\cite{anderson55}, we obtain an upper bound of the form
\begin{align}
p_S(\epsilon)&\leq \Pr\left[Z\in\cB_{\epsilon/\sqrt{\lambda_{\min}(\Sigma)}}(0)\right]\qquad\textrm{ where }\qquad Z\in\cN(0,I_k)\ , \label{eq:upperboundinequalitymain}
\end{align}
with $\lambda_{\min}(\Sigma)$ denoting the smallest eigenvalue of~$\Sigma$, see Corollary~\ref{cor:upperboundpsepsilon}. This appears rather crude, especially as it does not depend on~$\Delta$. For our purposes, however, this bound turns out to be sufficient.

\subsection{Decoding success probability and displaced Voronoi cells\label{sec:decsuccvoronoi}}
Let $X\sim\cN(0,\Sigma)$ be as before and consider the function~$\hat{x}^{\mathsf{MAP}}:[-\Delta/2,\Delta/2)^{n-k}\rightarrow\mathbb{R}^n$ defined by~\eqref{eq:mapestimatortwo}. Let $u\in [-\Delta/2,\Delta/2)^{n-k}$. 
Since every $y\in\mathbb{R}^n$ with $y_B^*=u$ is of the form $y=\binom{h}{u+\Delta b}$
for some $(h,b)\in \mathbb{R}^k\times \mathbb{Z}^{n-k}$, we can write
\begin{align}
\hat{x}^{\mathsf{MAP}}(u)&=
\binom{h(u)}{u+\Delta b(u)}\qquad\textrm{ where }\qquad 
\binom{h(u)}{b(u)}:=\arg\max_{\binom{h}{b}\in\mathbb{R}^k\times\mathbb{Z}^{n-k}}
f_\Sigma\left(\begin{pmatrix}
h\\
u+\Delta b
\end{pmatrix}\right)\ .\label{eq:defmaximumlikelihood}
\end{align}
The decoder $h^{\mathsf{MAP}}:=\Pi_A\circ \hat{x}^{\mathsf{MAP}}$ is then given by the function $h$ and has success probability 
\begin{align}
p_S(\epsilon)&=\Pr\left[X\in \cR(\epsilon)\right]\qquad\textrm{ with }\qquad \cR(\epsilon):=\left\{x\in \mathbb{R}^n\ \big|\ \|\Pi_Ax-h^{}(x_B^*)\|\leq \epsilon\right\}\ .
\end{align}
The set $\cR(\epsilon)$ of realizations~$x$ of the random vector~$X$ which lead to decoding success  can be expressed as follows.

\begin{lemma}\label{lem:voronoiset} Let $\epsilon>0$. Then 
$x\in\cR(\epsilon)$ if and only if $z\coloneqq \Sigma^{-1/2}x$ satisfies the following. There is some $c\in \cB_\epsilon(0)\times \Delta\mathbb{Z}^{n-k}$ such that
\begin{align}
\|z-\Sigma^{-1/2}c\|\leq \left\|z-\Sigma^{-1/2}\binom{h}{\Delta b}\right\|\qquad\textrm{ for all }\qquad (h,b)\in\mathbb{R}^k\times\mathbb{Z}^{n-k}\ . \label{eq:zcondtition}
\end{align}
\end{lemma}
\begin{proof}
Suppose that~$x\in\cR(\epsilon)$ and $x_B^*$ is defined by~\eqref{eq:xBstardefinition}. 
By the definitions~\eqref{eq:defmaximumlikelihood}, \eqref{eq:probdensfctsigma}, we have 
\begin{align}
\begin{pmatrix}
h(x_B^*)\\
b(x_B^*)
\end{pmatrix}&=\arg\max_{\binom{h}{b}\in\mathbb{R}^k\times\mathbb{Z}^{n-k}}
f_\Sigma\left(\begin{pmatrix}
h\\
x_B^*+\Delta b
\end{pmatrix}\right)\\
&=\arg\min_{\binom{h}{b}\in\mathbb{R}^k\times\mathbb{Z}^{n-k}}
 \begin{pmatrix}
h\\
x_B^*+\Delta b
\end{pmatrix}^T \Sigma^{-1}\begin{pmatrix}
h\\
x_B^*+\Delta b
\end{pmatrix}\\
&=\arg\min_{\binom{h}{b}\in\mathbb{R}^k\times\mathbb{Z}^{n-k}}
 \left\|
 \Sigma^{-1/2}
 \begin{pmatrix}
h\\
x_B^*+\Delta b
\end{pmatrix}\right\|\ ,\label{eq:minexpr}
\end{align}
where the last equation follows by symmetry of the covariance matrix $\Sigma$.
Eq.~\eqref{eq:minexpr} and the definitions of~$\cR(\epsilon)$ and $h(x_B^*)$ imply that 
$x\in\cR(\epsilon)$ if and only if there is some $h_0'\in\cB_\epsilon(\Pi_Ax)$ and $b_0'\in\mathbb{Z}^{n-k}$ such that
\begin{align}
\left\|
 \Sigma^{-1/2}
 \begin{pmatrix}
h_0'\\
x_B^*+\Delta b_0'
\end{pmatrix}\right\| & \leq \left\|
 \Sigma^{-1/2}
 \begin{pmatrix}
h'\\
x_B^*+\Delta b'
\end{pmatrix}\right\|\qquad\textrm{ for all }(h',b')\in\mathbb{R}^k\times\mathbb{Z}^{n-k}\ .
\end{align}
By definition of $x_B^*$ as the modulo-$\Delta$ reduced vector~$\Pi_Bx$, we have the following: for  every $b'\in\mathbb{Z}^{n-k}$, there is some $b\in\mathbb{Z}^{n-k}$ such that $x_B^*+\Delta b'=\Pi_Bx+\Delta b$. We can apply this substitution to both~$b_0'$ and $b'$, yielding the existence of some $h_0'\in\cB_\epsilon(\Pi_Ax)$ and $b_0\in\mathbb{Z}^{n-k}$ such that 
\begin{align}
\left\|
 \Sigma^{-1/2}
 \begin{pmatrix}
h_0'\\
\Pi_Bx+\Delta b_0
\end{pmatrix}\right\| & \leq \left\|
 \Sigma^{-1/2}
 \begin{pmatrix}
h'\\
\Pi_Bx+\Delta b
\end{pmatrix}\right\|\qquad\textrm{ for all }(h',b)\in\mathbb{R}^k\times\mathbb{Z}^{n-k}\ .
\end{align}
Let us write $h_0'=\Pi_Ax+h_0$ with $h_0\in\cB_\epsilon(0)$ and similarly $h'=\Pi_Ax+h$. We then conclude that there exists $h_0\in\cB_\epsilon(0)$ and $b_0\in \mathbb{Z}^{n-k}$ such that
\begin{align}
\left\|
 \Sigma^{-1/2}
 \begin{pmatrix}
\Pi_Ax+h_0\\
\Pi_Bx+\Delta b_0
\end{pmatrix}\right\| & \leq \left\|
 \Sigma^{-1/2}
 \begin{pmatrix}
\Pi_Ax+h\\
\Pi_Bx+\Delta b
\end{pmatrix}\right\|\qquad\textrm{ for all }(h,b)\in\mathbb{R}^k\times\mathbb{Z}^{n-k}\ ,
\end{align}
or equivalently
\begin{align}
\left\|z+\Sigma^{-1/2}\binom{h_0}{\Delta b_0}\right\| \leq  \left\|z+\Sigma^{1/2}\binom{h}{\Delta b}\right\|\qquad\textrm{ for all }(h,b)\in\mathbb{R}^k\times\mathbb{Z}^{n-k}\ .
\end{align}
The claim follows because $\binom{h_0}{\Delta b_0}\in \cB_\epsilon(0)\times \Delta\mathbb{Z}^{n-k}$ by definition, and thus $-\binom{h_0}{\Delta b_0}\in\cB_\epsilon(0)\times \Delta\mathbb{Z}^{n-k}$ since this set is symmetric around the origin. 
\end{proof}

\begin{theorem}\label{thm:mainvoronoireformulation}
Let $\cV_{\Sigma^{-1/2}}$ be the Voronoi cell
of the (degenerate) lattice $\Sigma^{-1/2}(\mathbb{R}^k\times\Delta\mathbb{Z}^{n-k})$ defined by~\eqref{eq:voronoicelldef}. Let $\epsilon>0$. 
Then
\begin{align}
p_S(\epsilon)&=\Pr\left[Z\in \Sigma^{-1/2}(\cB_\epsilon(0)\times \Delta\mathbb{Z}^{n-k})+\cV_{\Sigma^{-1/2}}\right]\ ,
\end{align}
where $Z\sim\cN(0,I_n)$.
\end{theorem}
\begin{proof}
 Lemma~\ref{lem:voronoiset} states that
$z\in \Sigma^{-1/2}\cR(\epsilon)$ if and only if there is some $c\in\cB_\epsilon(0)\times\Delta\mathbb{Z}^{n-k}$ such that~\eqref{eq:zcondtition} holds. But since the inequality~\eqref{eq:zcondtition} holds for all $(h,b)\in\mathbb{R}^{k}\times\mathbb{Z}^{n-k}$ and $c$ is of the form $c=\binom{h_0}{\Delta b_0}$ with $(h_0,b_0)\in\mathbb{R}^k\times\mathbb{Z}^{n-k}$ by definition,  condition~\eqref{eq:zcondtition} is equivalent to
\begin{align}
\|z-\Sigma^{-1/2}c\|\leq \left\|(z-\Sigma^{-1/2}c)-\Sigma^{-1/2}\binom{h}{\Delta b}\right\|\qquad\textrm{ for all }\qquad (h,b)\in\mathbb{R}^k\times\mathbb{Z}^{n-k}\ .
\end{align}
This shows that $z\in\Sigma^{-1/2}\cR(\epsilon)$ if and only if $z-\Sigma^{-1/2}c\in\cV_{\Sigma^{-1/2}}$ for some $c\in \cB_\epsilon(0)\times\Delta\mathbb{Z}^{n-k}$, i.e.,
\begin{align}
\Sigma^{-1/2}\cR(\epsilon)&=\cV_{\Sigma^{-1/2}}+\Sigma^{-1/2}(\cB_\epsilon(0)\times\Delta\mathbb{Z}^{n-k})\ .\label{eq:transformedsetofrecoverablevectors}
\end{align}
Since $Z\coloneqq \Sigma^{-1/2}X\sim \cN(0,I_n)$ and $p_s(\epsilon)=\Pr\left[X\in\cR(\epsilon)\right]$, this implies the claim.
\end{proof}
\begin{figure}
\centering
\subfigure[Illustration of an uncountable number of half-spaces defining~$\cV_\Lambda$.
Every pair $(h,b)\in\mathbb{R}^k\times \mathbb{Z}^{n-k}$ defines a half-space~$\cH_{(h,b)}$ (cf.~\eqref{eq:hhb}). Half-spaces associated to pairs $(h,0)$ with $b=0$ are shown in red.\label{subfigacountable}]{\includegraphics[width=8cm]{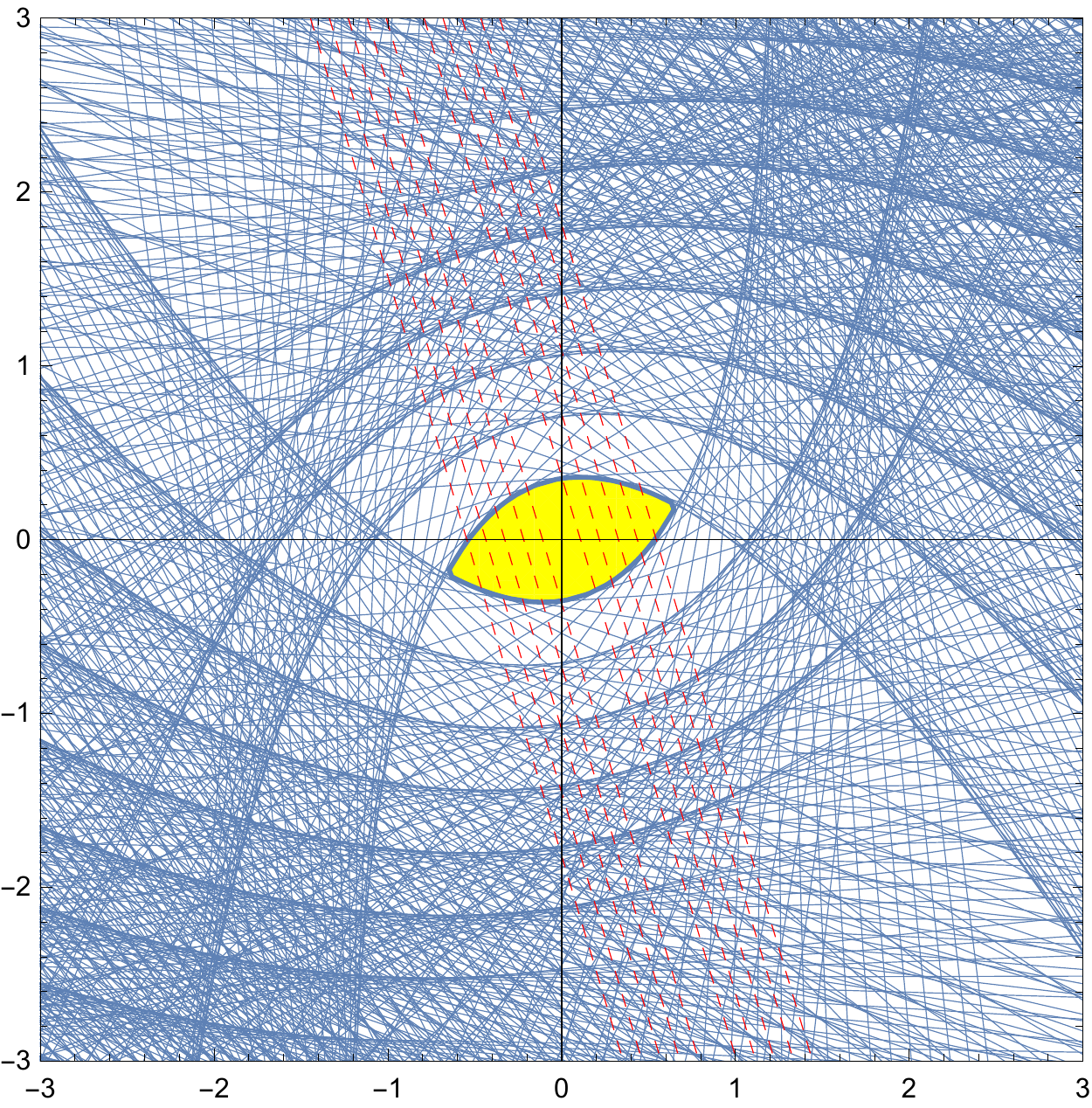}}\quad
\subfigure[A countable number of surfaces defined by quadratic constraints.
 Every $b\in\mathbb{Z}^{n-k}\backslash \{0\}$ defines a set~$\cE_b$ bounded by a surface defined by a quadratic form (cf.~\eqref{eq:cebcdef}).  The set~$\cE_0$ is a hyperplane (red).\label{secfigbcount}]{\includegraphics[width=8cm]{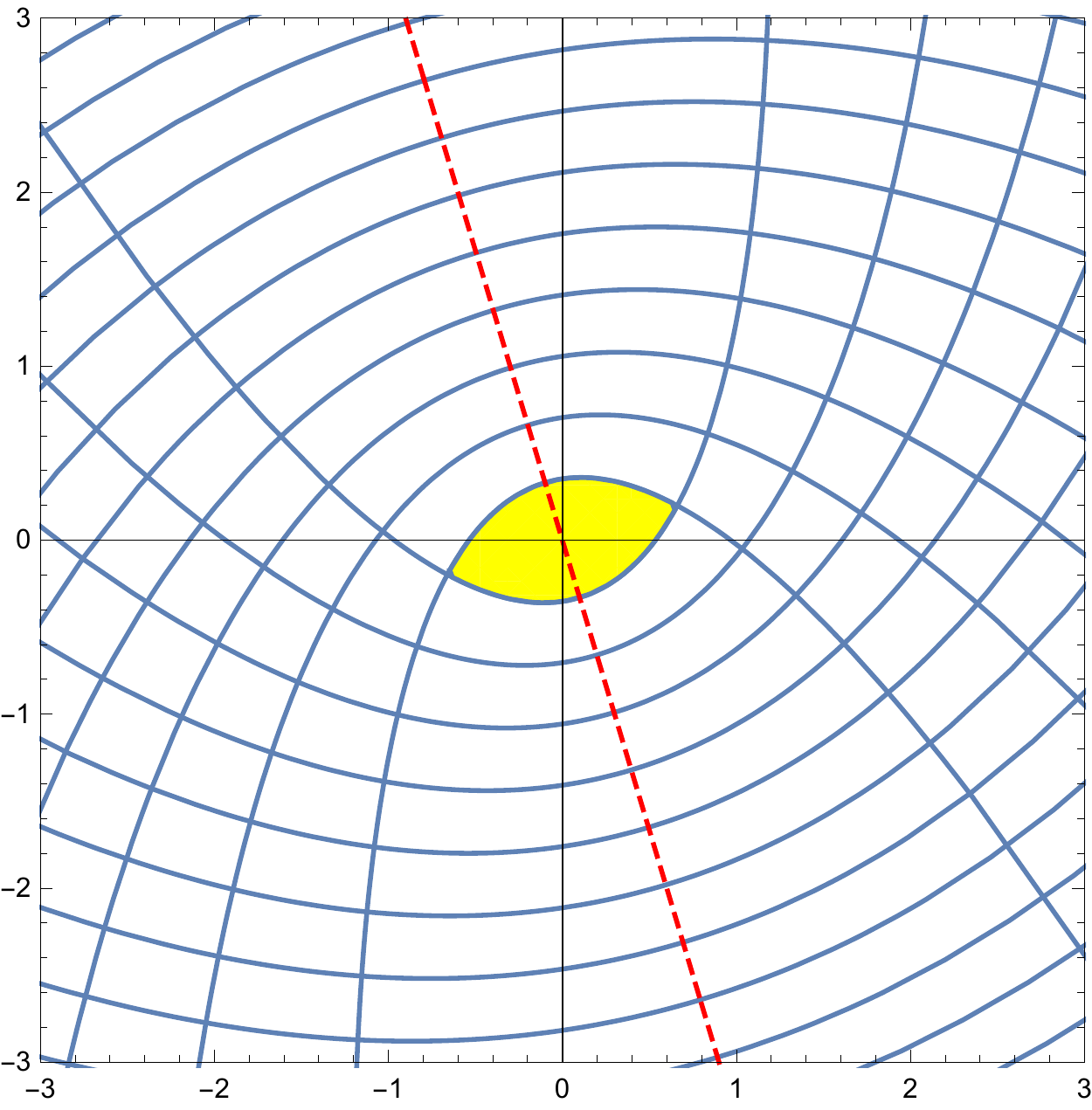}}
\caption{Different interpretations for the degenerate Voronoi cell~$\cV_\Lambda$. The set~$\cV_\Lambda$ is the intersection of the yellow lemon-shaped region and the red line in Fig.~\ref{secfigbcount}. }
\end{figure}

\subsection{A polytope containing the (transformed) Voronoi cell}\label{subsec:cp:propdegvoronoi}

To derive an upper bound on $p_S(\epsilon)$, we need detailed information about the Voronoi cell~$\cV_{\Sigma^{-1/2}}$ and the union of its thickened displaced versions (cf.\ \eqref{eq:transformedsetofrecoverablevectors}). For brevity, we use the shorthand $\Lambda\coloneqq \Sigma^{-1/2}$ in the following.  Recall that this matrix is positive definite and symmetric. The set of interest is
\begin{align}
\cV_{\Lambda}\coloneqq \left\{z\in\mathbb{R}^n\ \left|\ \|z\|\leq \left\|z-\Lambda \binom{h}{\Delta b}\right\| \textrm{ for all }(h,b)\in\mathbb{R}^k\times\mathbb{Z}^{n-k}\right.\right\}\ .\label{eq:vLambdainterest}
\end{align}
This is the Voronoi cell of the degenerate lattice $\Lambda(\mathbb{R}^k\times \Delta\mathbb{Z}^{n-k})$. It can be thought of in several different ways: on the one hand, it is the intersection  of an uncountable number of half-spaces
\begin{align}
\cH_{(h,b)}\coloneqq \left\{z\in\mathbb{R}^n\ \left|\ \left\langle \Lambda\binom{h}{\Delta b},z\right\rangle \leq \frac{1}{2} \left\|\Lambda \binom{h}{\Delta b}\right\|^2\right.\right\}\qquad\textrm{ with }(h,b)\in\mathbb{R}^{k}\times\mathbb{Z}^{n-k}\ .\label{eq:hhb}
\end{align}
This is illustrated in Fig.~\ref{subfigacountable}. 

Alternatively, the set $\cV_{\Lambda}$ can be understood as the intersection of a countable number of sets of the form
\begin{align}
\cE_{b}\coloneqq \left\{z\in\mathbb{R}^{n}\ \left|\ \left\langle\Lambda\binom{h}{\Delta b},z\right\rangle \leq \frac{1}{2} \left\|\Lambda \binom{h}{\Delta b}\right\|^2\textrm{ for all }h\in\mathbb{R}^k\right.\right\}\qquad\textrm{ with } b\in\mathbb{Z}^{n-k}\ , \label{eq:cebcdef}
\end{align}
 see Fig.~\ref{secfigbcount} for an illustration. 
Each set $\cE_b$ with $b\neq 0$ is of the form $\cE_b:=\{z\in\mathbb{R}^{n}\ |\ q_b(z)\geq 0\}$ for a quadratic form $q_b:\mathbb{R}^n\rightarrow\mathbb{R}$.  The explicit form of $q_b$ is given in Lemma~\ref{lem:quadraticformb} in Appendix~\ref{app:quadraticform}. The set $\cE_0$ is specified by linear equality constraints, see Lemma~\ref{lem:usefulpolytope} below.

Our goal here is not to provide a full characterization of this set. Instead, we establish a few necessary conditions for $v\in\mathbb{R}^n$ to belong to this set. To state these conditions, let $\langle x,y\rangle=\sum_{j=1}^n x_jy_j$ denote the standard inner product on~$\mathbb{R}^n$, and let $\{e_j\}_{j=1}^n$ be the canonical orthonormal basis of~$\mathbb{R}^n$.

Let $\hat{\Lambda}\in\mathsf{Mat}_{k\times (n-k)}(\mathbb{R})$ be an arbitrary matrix (to be chosen later).
Consider the lattice $\cL$ generated by the matrix 
\begin{align}
L(\Lambda,\hat{\Lambda},\Delta):=\Lambda\cdot \begin{pmatrix}
I & \Delta\hat{\Lambda}\\
0 & \Delta I
\end{pmatrix}\ ,\label{eq:Llambdadef}
\end{align}
i.e., $\cL=L\mathbb{Z}^{n}$. (Note that $L$ has non-zero determinant, whence $\cL$ is well-defined.)
Let $\cV(\cL)$ denote the Voronoi cell of $\cL$. 
By definition, $v\in\cV(\cL)$ if and only if 
\begin{align}
|\langle x, v\rangle|\leq \frac{\|x\|^2}{2}\qquad\textrm{ for every }x\in\cL\ . \label{eq:latticevoronoidef} 
\end{align}
The following shows that the Voronoi cell~$\cV(\cL)$ of~$\cL$ contains the set~$\cV_\Lambda$ of interest.

\begin{lemma}\label{lem:usefulpolytope}
Let $v\in \cV_\Lambda$. Then
\begin{align}
\langle \Lambda e_j,v\rangle &=0\qquad\textrm{ for }j=1,\ldots,k\label{eq:voronoicoeff}
\end{align}
and
\begin{align}
\cV_\Lambda\subset \cV(\cL)\ .\label{eq:voronoicellinclusion}
\end{align}
\end{lemma}

\begin{proof}
Let $v\in\cV_\Lambda$ and $j=1,\ldots,k$. We have 
\begin{align}
\|v\|&\leq \|v-\Lambda \delta e_j\|\qquad\textrm{ for all }\delta\in\mathbb{R}\ 
\end{align}
by choosing $\binom{h}{\Delta b}=\binom{h}{0}=\delta e_j$ accordingly in Eq.~\eqref{eq:vLambdainterest}. 
This implies~\eqref{eq:voronoicoeff}. 

Similarly, choosing $\binom{h}{\Delta b}=\binom{a+\hat{\Lambda} \Delta b}{\Delta b}$ with $(a,b)\in\mathbb{Z}^k\times\mathbb{Z}^{n-k}$ yields the condition
\begin{align}
\|v\|\leq \left\|v-\Lambda\binom{a+\hat{\Lambda}\Delta b}{\Delta b}\right\|\qquad\textrm{ for all }(a,b)\in\mathbb{Z}^n\ 
\end{align}
for any $v\in\cV_\Lambda$. 
This is equivalent to
\begin{align}
\|v\|\leq \left\|v-x\right\|\qquad\textrm{ for all }x\in\cL\ ,
\end{align}
and thus~\eqref{eq:voronoicellinclusion} follows. 
\end{proof}

Let $\cP_\Lambda$ be the set of $v\in\cV(\cL)$ which satisfy~\eqref{eq:voronoicoeff}.

\begin{figure}
\centering
\includegraphics[width=8cm]{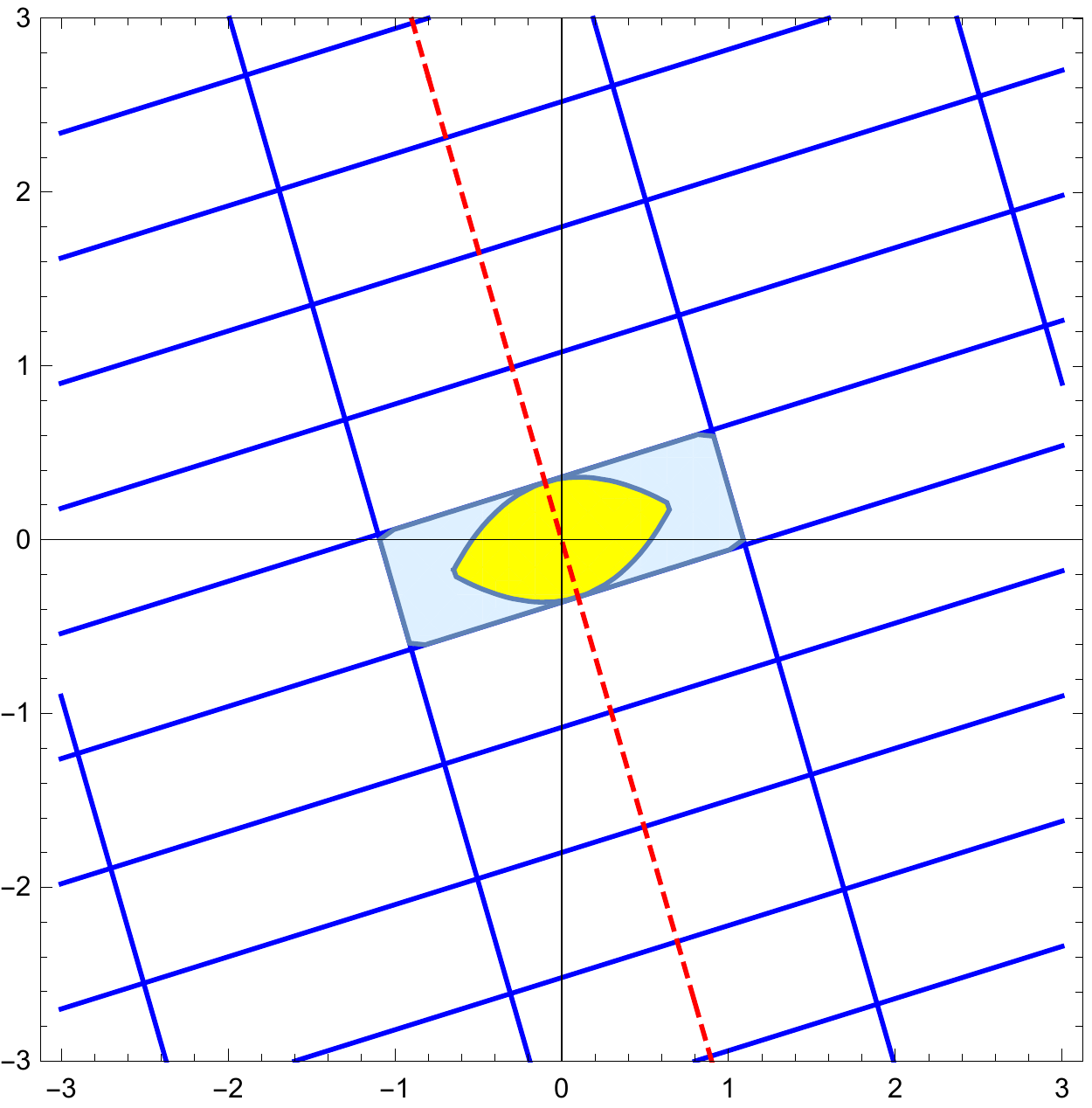}
\caption{Translates of the Voronoi cell of the lattice~$L\mathbb{Z}^{n}$, the degenerate Voronoi cell~$\cV_\Lambda$ (intersection of the yellow lemon-shaped region with the red straight line), and the polytope~$\cP_\Lambda$ (intersection of the rectangular blue region with the red straight line). }
\end{figure}

\begin{lemma}\label{lem:polytope}
$\cP_\Lambda$ is an $(n-k)$-dimensional polytope.
\end{lemma}
\begin{proof}
The Voronoi cell~$\cV(\cL)$ is a  polytope, i.e., the intersection of a finite number of halfspaces. Furthermore, it is an $n$-dimensional polytope. (For a detailed proof of these and related facts, see e.g.,~\cite[Chapter 4]{coopel2009}.) The definition of the set~$\cP_\Lambda$ additionally includes the linear equality constraints~\eqref{eq:voronoicoeff} (which are all independent since the matrix~$\Lambda$ is invertible), hence the claim follows.
\end{proof}
 By Lemma~\ref{lem:usefulpolytope}, we have~$\cV_\Lambda\subset \cP_\Lambda$, i.e., $\cV_\Lambda$ is contained in the $(n-k)$-dimensional polytope $\cP_\Lambda$. We will use $\cP_\Lambda$ as a proxy for~$\cV_\Lambda$ in the following.

\subsection{Parametrization of~$\cP_\Lambda$ and choice of $\hat{\Lambda}$\label{sec:parametrizationPLambda}}
Lemma~\ref{lem:polytope} shows that  $n-k$~real parameters suffice to parametrize the set~$\cP_\Lambda$.  We now make this parametrization explicit. To do so, it is convenient to make a specific choice of $\hat{\Lambda}\in \mathsf{Mat}_{k\times (n-k)}(\mathbb{R})$ (see Eq.~\eqref{eq:hatlambdadefinition} below). 

It is easy to show (using a proof analogous to that of Lemma~\ref{lem:sigmahattwo} below) that every element $v\in\cP_\Lambda$ is uniquely determined by its projection~$\Pi_B v$.
For our purposes, it is more convenient to derive such a statement not for~$\cP_\Lambda$, but for the transformed set~$\Lambda^{-1}\cP_\Lambda$.

To state this result, we will henceforth identify linear maps $M:\mathbb{R}^n\rightarrow\mathbb{R}^n$ with their associated matrix $M_{i,j}=\langle e_i,M e_j\rangle$ when expressed in the standard basis. We also define block matrices corresponding to the partition~$\mathbb{R}^n\cong \mathbb{R}^k\times\mathbb{R}^{n-k}$, i.e., we write
\begin{align}
M&=\begin{pmatrix}
M_{AA}&M_{AB}\\
M_{BA}&M_{BB}
\end{pmatrix}\ .
\end{align}

\begin{lemma}\label{lem:sigmahattwo}
Let $\hat{\Lambda}:\mathbb{R}^{n-k}\rightarrow\mathbb{R}^k$ be defined as 
\begin{align}
\hat{\Lambda}=- ((\Lambda^2)_{AA})^{-1}(\Lambda^2)_{AB}\ .\label{eq:hatlambdadefinition}
\end{align}
Then 
\begin{align}
(\pA-\hat{\Lambda}\pB)\Lambda^{-1}\cP_\Lambda=\{0\}\ .\label{eq:parametrizationequationtwo}
\end{align}
\end{lemma}
\begin{proof}
Let $w\in \Lambda^{-1}\cP_\Lambda$. Then  $\Lambda w\in\cP_\Lambda$ and thus 
\begin{align}
\langle \Lambda e_j,\Lambda w\rangle &=0\qquad\textrm{ for all }j=1,\ldots ,k\label{eq:zerocondatwo}
\end{align}
by definition of $\cP_\Lambda$. Using the assumption that $\Lambda^T=\Lambda$ is symmetric, we have 
\begin{align}
\langle \Lambda e_j,\Lambda w\rangle&= \langle e_j, \Lambda^2 w\rangle\\
&= \langle e_j, \Pi_A \Lambda^2 w\rangle\\
&=\langle e_j, (\Lambda^2)_{AA} \Pi_A w+(\Lambda^2)_{AB}\Pi_B w\rangle\ 
\end{align}
for all $j=1,\ldots,k$. 
Inserted into~\eqref{eq:zerocondatwo}, we conclude that
\begin{align}
(\Lambda^2)_{AA} \Pi_A w&=-(\Lambda^2)_{AB}\Pi_B w\ .
\end{align}
The matrix $(\Lambda^2)_{AA}$ is a principal submatrix of  the positive definite matrix~$\Lambda^2$, hence it is also positive definite and thus invertible. Since~$w\in \Lambda^{-1}\cP_\Lambda$ was arbitrary, this implies that 
\begin{align}
\pA w&= \hat{\Lambda}\pB w\qquad\textrm{ for all }\qquad w\in\Lambda^{-1}\cP_\Lambda\ ,
\end{align}
which is the claim.
\end{proof}

Eq.~\eqref{eq:parametrizationequationtwo} of Lemma~\ref{lem:sigmahattwo} shows that the map
\begin{align}
\begin{matrix}
\varphi: & \Pi_B\Lambda^{-1}\cP_\Lambda & \rightarrow &  \Lambda^{-1}\cP_\Lambda\\
& z & \mapsto & \binom{\hat{\Lambda} z}{z}\ 
\end{matrix}\label{eq:varphimap}
\end{align}
is bijective. Thus we have obtained a bijective parametrization of the transformed polytope~$\Lambda^{-1}\cP_\Lambda$.

\subsection{Almost-unique parametrization of translates\label{sec:translationparametrization}}
The following is well-known (see e.g.,~\cite[Chapter 4]{coopel2009}): If $\cV$ is the Voronoi cell of the lattice $\Lambda \mathbb{Z}^n$, then the lattice translates 
$\Lambda c+ \cV$, $c\in\mathbb{Z}^n$ are  a tiling of $\mathbb{R}^n$, i.e.,
\begin{align}
\mathbb{R}^n &= \bigcup_{c\in \mathbb{Z}^n}  (\Lambda c+ \cV)\ ,\\
\mathsf{int}\left(\Lambda c+ \cV\right)\cap \mathsf{int} \left(\Lambda c'+\cV\right)&=\emptyset \qquad\textrm{ if }c,c'\in\mathbb{Z}^n\textrm{ and }c\neq c'\ ,
\end{align}
where $\mathsf{int}(M)$ denotes the interior of the set~$M$. 
Furthermore, the intersection $\left(\Lambda c+ \cV\right)\cap  \left(\Lambda c'+ \cV\right)$ (if non-empty) 
is a face of both $\Lambda c+ \cV$ and $\Lambda c'+ \cV$. We show a similar statement for the polytope~$\cP_\Lambda$, again considering $\Lambda^{-1}\cP_\Lambda$ for convenience. In particular, we give a partial characterization of the intersection of two translates. 

In more detail, our goal is to give a parametrization of the union of translates $\binom{0}{\Delta b}+ \Lambda^{-1}\cP_\Lambda$ with $b\in\mathbb{Z}^{n-k}$
of the transformed polytope $\Lambda^{-1}\cP_\Lambda$. Since the latter can be uniquely parametrized by the projected set $\Pi_B\Lambda^{-1}\cP_\Lambda$ (see Eq.~\eqref{eq:varphimap}), we would like to show that projecting the translates also leads to disjoint sets. We show that this is essentially the case up to zero-measure sets, see Lemma~\ref{lem:translateintersection} below.

As a preparation, let us first consider translations of the transformed polytope $\Lambda^{-1}\cP_\Lambda$ along certain lattice directions of the lattice~$\Lambda^{-1}\cL$.
Recall that by definition of~$\cL$ (cf.~\eqref{eq:Llambdadef}), the lattice~$\Lambda^{-1}\cL$ consists of points $\binom{a+\hat{\Lambda }\Delta b}{\Delta b}$ with $(a,b)\in\mathbb{Z}^k\times\mathbb{Z}^{n-k}$. 

\begin{lemma}\label{lem:nonstandardtranslate}
For $b\in\mathbb{Z}^{n-k}$ set
\begin{align}
\cU_b\coloneqq  \binom{\hat{\Lambda}\Delta b}{\Delta b}+\Lambda^{-1}\cP_\Lambda\ .
\end{align}
Let $b,b'\in\mathbb{Z}^{n-k}$ with $b\neq b'$. Then
\begin{align}
\Pi_B(\cU_b\cap \cU_{b'})
\end{align}
is contained in a polytope of dimension~$(n-k-1)$.
\end{lemma}
 \begin{proof}
 By considering the difference of $b$ and $b'$, we can assume that~$b\neq 0$ and $b'=0$ without loss of generality.
 Suppose that $ w\in\cU_b\cap \cU_0$. 
 The fact that $w\in\cU_b$ implies that there is some $w'\in\Lambda^{-1}\cP_\Lambda$ such that
\begin{align}
w&=w'+\Lambda^{-1 }x\qquad\textrm{ where }\qquad x=\Lambda\binom{\hat{\Lambda}\Delta b}{\Delta b}\in\cL\ .\label{eq:xdefinitionh}
\end{align}
Observe that $x\neq 0$ as $b\neq 0$ and the (linear) map that takes $b$ to $x$ is invertible. 
By Eq.~\eqref{eq:latticevoronoidef} and  because $w,w'\in\Lambda^{-1}\cP_\Lambda$ we have 
\begin{align}
    \|x\|^{-2} \cdot \left|\langle x,\Lambda w\rangle \right|&\leq 1/2\ \label{eq:secondinequality}\\
   \|x\|^{-2} \cdot \left|\langle x,\Lambda w'\rangle \right|&\leq 1/2\label{eq:firstequality}\ .
\end{align}
Suppose  that $\|x\|^{-2}\cdot \langle x,\Lambda w'\rangle >-1/2$. Then
\begin{align}
  \|x\|^{-2} \cdot \langle x,\Lambda w\rangle
  &=   \|x\|^{-2} \cdot \langle x,\Lambda w'+x\rangle\\
  &=  \|x\|^{-2} \cdot \langle x,\Lambda w'\rangle+1\label{eq:xm}\\
  &>-1/2+1 =1/2\ ,
\end{align}
contradicting~\eqref{eq:secondinequality}.  It follows that we must have 
\begin{align}
 \| x\|^{-2} \cdot\langle  x,\Lambda w'\rangle&=-1/2\ ,
 \end{align}
 and with Eq.~\eqref{eq:xm} also 
\begin{align}
 \| x\|^{-2} \cdot \langle  x,\Lambda w\rangle&=1/2\ .
 \end{align}
In particular, every $w\in \cU_b\cap \cU_0$ satisfies the linear constraint
 \begin{align}
 \ell(w)=c\qquad\textrm{ where }\qquad \ell(w)\coloneqq \langle x,\Lambda w\rangle\textrm{ and }c\coloneqq \|x\|^{2}/2\ ,\label{eq:linearconstraintellwc}
 \end{align}
 with $x$ defined by~\eqref{eq:xdefinitionh}. 
 
 Let us now parametrize $\cU_b\cap \cU_0\subset \Lambda^{-1}\cP_\Lambda$ by the map~$\varphi$ (cf.~\eqref{eq:varphimap}), i.e., we write $w=\varphi(z)$ for some $z\in\Pi_B\Lambda^{-1}\cP_\Lambda$. Because 
\begin{align}
\ell(w)&=\left\langle \Lambda \binom{\hat{\Lambda}\Delta b}{\Delta b},\Lambda w\right\rangle\ ,
\end{align}
we have
\begin{align}
\ell(\varphi(z))&=\Delta\cdot \left\langle \Lambda \binom{\hat{\Lambda} b}{ b},\Lambda \binom{\hat{\Lambda}z}{z}\right\rangle\\
&=\Delta\cdot  \langle K b, Kz\rangle\ ,
\end{align}
where we introduced the map
\begin{align}
\begin{matrix}
K\colon &\mathbb{R}^{n-k}& \rightarrow\mathbb{R}^n\\
 &b&\mapsto &\Lambda \binom{\hat{\Lambda}b}{b}\ 
\end{matrix}\ .
\end{align}
The map~$K$ is invertible with  inverse $K^{-1}=\Pi_B\circ\Lambda^{-1}$ and thus has full rank.
It follows that $K^T K$ is positive definite. With Eq.~\eqref{eq:linearconstraintellwc}, we have thus obtained a non-trivial  linear constraint of the form
\begin{align}
\langle K^T K b,z\rangle=c/\Delta \qquad\textrm{ for any }\qquad z\in \varphi^{-1}(\cU_b\cap\cU_0)\subset \Pi_B \Lambda^{-1}\cP_{\Lambda}\ .\label{eq:mlinearcons}
\end{align}
In other words, $\varphi^{-1}(\cU_b\cap\cU_0)$ is the intersection of the $n-k-1$-dimensional hyperplane, described by the linear constraint~\eqref{eq:mlinearcons}, with $\Pi_B \Lambda^{-1}\cP_{\Lambda}$. Noting that linear transformations, as well as intersections with hyperplanes, map polytopes to polytopes, we conclude that $\varphi^{-1}(\cU_b\cap\cU_0)$ is contained in a polytope of dimension $n-k-1$. Since $\Pi_B w=\Pi_B \varphi(z)=z\in \varphi^{-1}(\cU_b\cap\cU_0)$ for $w\in \cU_b\cap\cU_0$, the same conclusion holds for $\Pi_B(\cU_b\cap\cU_0)\subset\mathbb{R}^{n-k}$.
 \end{proof}

We now consider  translations with respect to the standard rectangular grid instead of the lattice~$\cL$ used in the definition of~$\cP_\Lambda$.
\begin{lemma}\label{lem:translateintersection}
For $b\in\mathbb{Z}^{n-k}$ 
define 
\begin{align}
\cS_b&\coloneqq\binom{0}{\Delta b} +\Lambda^{-1}\cP_\Lambda\ .\label{eq:sbdefinition}
\end{align}
Let $b,b'\in\mathbb{Z}^{n-k}$ with $b\neq b'$. Then
\begin{align}
\Pi_B(\cS_b)\cap \Pi_B(\cS_{b'})
\end{align}
 is contained in a polytope of dimension~$(n-k-1)$.
 \end{lemma}

\begin{proof} 
Without loss of generality (cf.\ the proof of Lemma~\ref{lem:nonstandardtranslate}), assume that $b\neq 0$ and $b'=0$. Suppose
\begin{align}
z\in \Pi_B\cS_b\cap \Pi_B\cS_0\ .
\end{align}
Since $z\in \Pi_B\cS_0$ we have $\varphi(z)\in\cS_0$, i.e.,
\begin{align}
\binom{\hat{\Lambda}z}{z}\in\Lambda^{-1}\cP_\Lambda\ .
\end{align}
Since $z\in\Pi_B\cS_b$ we have 
\begin{align}
\varphi(z-\Delta b)\in\Lambda^{-1}\cP_\Lambda\ ,
\end{align}
that is,
\begin{align}
\binom{\hat{\Lambda}(z-\Delta b)}{z-\Delta b}\in\Lambda^{-1}\cP_\Lambda\ .
\end{align}
With $v\coloneqq \binom{\hat{\Lambda}z}{z}$ it follows that
\begin{align}
v\in\Lambda^{-1}\cP_\Lambda\qquad\textrm{ and }\qquad v-\binom{\hat{\Lambda}\Delta b}{\Delta b}\in \Lambda^{-1}\cP_\Lambda\ ,
\end{align}
or 
\begin{align}
v\in \Lambda^{-1}\cP_\Lambda \cap \left(\binom{\hat{\Lambda}\Delta b}{\Delta b}+\Lambda^{-1}\cP_\Lambda\right)=\cU_0\cap \cU_b\ .
\end{align}
With Lemma~\ref{lem:nonstandardtranslate}, we conclude that
$\Pi_Bv=z$ belongs to a polytope of dimension $(n-k-1)$, as claimed.
\end{proof}

Our goal is to upper bound integrals over the union~$\bigcup_{b\in\mathbb{Z}^{n-k}}\cS_b$ of translates~$\cS_b$ of the set~$\Lambda^{-1}\cP_\Lambda$. To this end, we are going to express~$\bigcup_{b\in\mathbb{Z}^{n-k}}\cS_b$   as the disjoint union of a countable number of measurable sets for which we have a unique parametrization (due to Lemma~\ref{lem:translateintersection}), and some measure-zero set (see Lemma~\ref{lem:zeromeasuredecomposition}). We will need the following lemma, which shows that any translate intersects only with a finite number of other translates.
 \begin{lemma}\label{lem:finiteintersection}
 For $b\in\mathbb{Z}^{n-k}$, let $\cS_b\subset\mathbb{R}^n$ be defined by~\eqref{eq:sbdefinition}. Then the following holds:
For any $b\in\mathbb{Z}^{n-k}$, there are only finitely many $b'\in\mathbb{Z}^{n-k}$ such that~$\cS_b\cap \cS_{b'}\neq \emptyset$.
  \end{lemma}
  \begin{proof}
  Let $b\in\mathbb{Z}^{n-k}$ be given. If $v\in\cS_{b}\cap \cS_{b'}$ for some $b'\in\mathbb{Z}^{n-k}$ then there
  are $w,w'\in\cP_\Lambda$ such that
  \begin{align}
  \binom{0}{\Delta b}+\Lambda^{-1} w&=\binom{0}{\Delta b'}+\Lambda^{-1}w'\ .
  \end{align}
  It follows that
  \begin{align}
  \|b-b'\| & = \|\Delta^{-1} \Lambda^{-1}(w-w')\|\\
  &\leq \lambda_{\max}\left(\Delta^{-1} \Lambda^{-1}\right)\cdot \|w-w'\|\\
  &\leq 2R\Delta^{-1}\lambda_{\max}\left( \Lambda^{-1}\right)
  \end{align}
  where $R:=\sup_{w\in\cP_\Lambda}\|w\|<\infty$ since $\cP_\Lambda$ is bounded. 
  Thus $b'\in\mathbb{Z}^{n-k}$ is contained in a ball of constant radius around $b$. The claim follows from this.
   \end{proof}

\newcommand*{\icS}{\mathring{\cS}}
We are interested in the union~$\bigcup_{b\in\mathbb{Z}^{n-k}}\cS_b$ of the translates~$\cS_b$, $b\in\mathbb{Z}^{n-k}$. To bound integrals over this set, the following statement will be used:
\begin{lemma}\label{lem:zeromeasuredecomposition}
There are measurable subsets
\begin{align}
\icS_b &\subset \cS_b\ ,\\
\cT_b &\subset \cS_b\qquad\textrm{ for }\qquad b\in\mathbb{Z}^{n-k}\ ,
\end{align}
such that
\begin{align}
\left(\bigcup_{b\in\mathbb{Z}^{n-k}}\icS_b\right)\cup \left(\bigcup_{b\in\mathbb{Z}^{n-k}}\cT_b\right)&=\bigcup_{b\in\mathbb{Z}^{n-k}}\cS_b\ ,\\
\icS_b\cap \icS_{b'}&=\emptyset\qquad\textrm{ for }b\neq b'\ ,\\
\icS_b\cap \cT_{b'}&=\emptyset\qquad\textrm{ for all }b,b'\ ,
\end{align}
and each set $\Pi_B\cT_b$ is of measure zero (and thus so is $\Pi_B\left(\bigcup_{b\in\mathbb{Z}^{n-k}}\cT_b\right)$).
\end{lemma}
\begin{proof}
The set~$\mathbb{Z}^{n-k}$ is countable. Fix an enumeration of $\mathbb{Z}^{n-k}$. Slightly abusing notation, let us write $\{\cS_N\}_{N\in\mathbb{N}}$ instead of $\{\cS_b\}_{b\in\mathbb{Z}^{n-k}}$. 

Now define
\begin{align}
\cT_1&=\bigcup_{M>1}(\cS_1\cap \cS_M)\qquad\textrm{ and }\qquad \icS_1=\cS_1\backslash \cT_1\ ,\\
\cT_2&=\bigcup_{M>2}(\cS_2\cap \cS_M)\qquad\textrm{ and }\qquad \icS_2=\cS_2\backslash (\cT_1\cup \cT_2)\ ,
\end{align}
and more generally
\begin{align}
\cT_J&=\bigcup_{M>J} (\cS_J\cap \cS_M)\ ,\\
\icS_J&=\cS_J \setminus \left(\bigcup_{K\leq J} \cT_K\right)\ ,
\end{align}
for $J\in\mathbb{N}$. By definition, the sets $\icS_J$ are pairwise disjoint, $\icS_J\cap\cT_{J'}=\emptyset$ for all $J,J'\in\mathbb{N}$, and 
\begin{align}
\left(\bigcup_{N\in\mathbb{N}}\icS_N\right)\cup \left(\bigcup_{N\in\mathbb{N}}\cT_N\right)&=\bigcup_{N\in\mathbb{N}}\cS_N\ .
\end{align}
Furthermore, by Lemma~\ref{lem:finiteintersection}, each~$\cT_J$ is a union of a finite number of sets $\cS_J\cap\cS_M$, and therefore $\bigcup_{N\in\mathbb{N}}\cT_N$ is a countable union of such sets. Together with Lemma~\ref{lem:translateintersection}, this implies that $\Pi_B\cT_J$ ($\Pi_B\bigcup_{N\in\mathbb{N}}\cT_N$) is a finite (countable) union of sets $\Pi_B(S_J\cap S_M)\subset \Pi_B(S_J)\cap\Pi_B(S_M)$ contained in a polytope of dimension~$n-k-1$. Since any set of dimension~$<n-k$ has Lebesgue measure zero in~$\mathbb{R}^{n-k}$, and countable unions of zero-measure sets are also of measure zero, the claim follows. 
\end{proof}

\subsection{Integrating over the union of translated and thickened polytopes\label{sec:displacedthick}}
To obtain a bound on the decoding success probability, we first establish a general upper bound on the probability mass of 
the union of translated and thickened versions of the transformed polytope~$\Lambda^{-1}\cP_\Lambda$. This is expressed in Lemma~\ref{lem:mainintegrationlemmageneral} below. 

Let us first consider integrals over subsets of the set~$\cS_b$.
Recall that $\cB_\delta(0)\subset \mathbb{R}^k$ is the closed $\delta$-ball with respect to the Euclidean norm, centered at the origin. For the remainder of this paper, we denote by  
\begin{align}
\cS(\epsilon)\coloneqq \cS+\cB_\epsilon(0)\times \{0\}^{n-k}\label{eq:epsilonthickeningdef}
\end{align}
the $\epsilon$-thickening of a set~$\cS\subset\mathbb{R}^n$, for $\epsilon>0$. 

\begin{lemma}\label{lem:parametrizedintegral}
Let $b\in\mathbb{Z}^{n-k}$ be arbitrary.
Let $\cQ_b\subset\cS_b=\binom{0}{\Delta b}+\Lambda^{-1}\cP_\Lambda$ be a measurable subset. 
 Let $g\colon \mathbb{R}^n\rightarrow \mathbb{R}$ be integrable. Then
\begin{align}
\int_{\cQ_b(\epsilon)}g(x)d x &= \int_{\xi\in \Pi_B\cQ_b}\int_{h\in\cB_\epsilon(0)}  g\left(
\binom{\hat{\Lambda}(\xi-\Delta b)+h}{\xi}
\right) dhd\xi\ .
\end{align}
\end{lemma}
\begin{proof}

Define
\begin{align}
\begin{matrix}
\phi_b\colon &\cB_\epsilon(0) \times \Pi_B\cQ_b& \rightarrow &\cQ_b(\epsilon)\\
& \binom{h}{\xi} & \mapsto &  \binom{\hat{\Lambda}(\xi-\Delta b)+h}{\xi}\ .
\end{matrix}\label{eq:phibmap}
\end{align}
By definition (and the fact that the map~$\varphi$ defined by~\eqref{eq:varphimap} is bijective), this map is surjective: every element $v\in \cQ_b(\epsilon)$ has the form
$v=\binom{\hat{\Lambda} z+h}{z+\Delta b}$  with $h\in\cB_\epsilon(0)$, $z\in \Pi_B\Lambda^{-1}\cP_\Lambda$
 and thus $v=\phi_b\left(\binom{h}{\xi}\right)$
 with $\xi=z+\Delta b\in \Pi_B \cQ_b$. It is also easy to check that it is injective.  Furthermore, the Jacobi-matrix  of~$\phi_b$ is
$D\phi_b=\begin{pmatrix}
I_{k} & \hat{\Lambda}\\
0 & I_{n-k}
\end{pmatrix}$
and satisfies 
\begin{align}
\det D\phi_b&=1\ .\label{eq:determinantDphione}
\end{align}
The claim therefore follows from the change-of-variables formula for Lebesgue integrals.
\end{proof}
 
 Let $f:\mathbb{R}^n\rightarrow \mathbb{R}$ be a probability density function
 of an absolutely continuous probability measure on~$\mathbb{R}^n$. Writing $f=f_{X_AX_B}$, this 
factorizes as 
\begin{align}
f_{X_AX_B}(x_A,x_B)&=f_{X_B}(x_B)f_{X_A|X_B=x_B}(x_A)\qquad\textrm{ for }\qquad(x_A,x_B)\in\mathbb{R}^k\times\mathbb{R}^{n-k}\ ,
\end{align}
where $f_{X_B}$ is the density function associated with the second marginal, and $f_{X_A|X_B=x_B}$ denotes the density associated with the conditional distribution of $X_A$ given that $X_B=x_B$. Our main technical statement is the following:
 
 \begin{lemma}\label{lem:mainintegrationlemmageneral}
 Let $f:\mathbb{R}^n\rightarrow \mathbb{R}$ be a probability density function.
 Then
\begin{equation}
 \int_{\bigcup_{b\in\mathbb{Z}^{n-k}}\cS_b(\epsilon)}f(x)d x \leq \sum_{b\in\mathbb{Z}^{n-k}} \int_{\xi\in\Pi_B\icS_b}f_{X_B}(\xi)I(b,\xi) d\xi\ , \label{eq:fxdxintegr}
 \end{equation}
  where
     \begin{align}
     I(b,\xi)&\coloneqq \int_{h\in\cB_\epsilon(0)} f_{X_A|X_B=\xi}\left(\hat{\Lambda}(\xi-\Delta b)+h\right)dh\qquad\textrm{ for }\qquad(\xi,b)\in\mathbb{R}^{n-k}\times\mathbb{Z}^{n-k}\ .\label{eq:Ibxidef}
              \end{align}
 Furthermore, we have the inequality
 \begin{align}
 \sum_{b\in\mathbb{Z}^{n-k}} \int_{\xi\in\Pi_B\icS_b}f_{X_B}(\xi) d\xi&\leq 1\ .\label{eq:inequalitytoshowlast}
 \end{align}
  \end{lemma}
  \begin{proof}
  Since 
  \begin{align}
  \bigcup_{b\in\mathbb{Z}^{n-k}}\cS_b(\epsilon)&= \left(\bigcup_{b\in\mathbb{Z}^{n-k}}\icS_b(\epsilon)\right)\cup
  \left(\bigcup_{b\in\mathbb{Z}^{n-k}}\cT_b(\epsilon)\right)
  \end{align}
  by Lemma~\ref{lem:zeromeasuredecomposition} and definition~\eqref{eq:epsilonthickeningdef}, we obtain
  \begin{align}
   \int_{\bigcup_{b\in\mathbb{Z}^{n-k}}\cS_b(\epsilon)}f(x)dx &\leq I_1+I_2\ ,
  \end{align}
  with
  \begin{align}
  I_1\coloneqq \sum_{b\in\mathbb{Z}^{n-k}}\int_{\icS_b(\epsilon)}f(x)dx\qquad\textrm{ and }\qquad 
  I_2\coloneqq \sum_{b\in\mathbb{Z}^{n-k}}\int_{\cT_b(\epsilon)}f(x)dx\ .
    \end{align}
    By Lemma~\ref{lem:parametrizedintegral}, the right hand side of Eq.~\eqref{eq:fxdxintegr} coincides with $I_1$. It thus suffices to show that $I_2=0$.
    But (again by Lemma~\ref{lem:parametrizedintegral}), we have 
    \begin{align}
    \int_{\cT_b(\epsilon)}f(x)dx&=\int_{\xi\in \Pi_B\cT_b}f_{X_B}(\xi)
    I(b,\xi) d\xi\\
    &\leq \int_{\xi\in \Pi_B\cT_b}f_{X_B}(\xi) d\xi\\
    &=0\ ,
    \end{align}
    as $\Pi_B\cT_b$ is a set of measure zero, see Lemma~\ref{lem:zeromeasuredecomposition}. This establishes~\eqref{eq:fxdxintegr}.  
    
        Now consider~\eqref{eq:inequalitytoshowlast}.
   We have
  \begin{align}
1&\geq  \int_{\Pi_B\bigcup_{b\in\mathbb{Z}^{n-k}}\cS_b}f_{X_B}(\xi)d\xi\\
&=\sum_{b\in\mathbb{Z}^{n-k}}\int_{\Pi_B\cS_b}f_{X_B}(\xi)d\xi\ ,
\end{align}
    where we used the fact that the projected sets~$\{\Pi_B\cS_b\}_{b\in\mathbb{Z}^{n-k}}$ have only zero-measure pairwise intersection, see Lemma~\ref{lem:translateintersection}. The claim follows since $\icS_b\subset\cS_b$ for every $b\in\mathbb{Z}^{n-k}$. 
  \end{proof}

\subsection{An upper bound on the decoding success probability}\label{sec:upperboundunwrapp}
We now return to the partial informed unwrapping problem (see Section~\ref{sec:partialinformedunwrapping}) and give upper bounds on the success probability~$p_S(\epsilon)$. Our central result is the following theorem previously summarized in Eq.~\eqref{eq:psepsilonupperboundgeneral}.
It is obtained by specializing Lemma~\ref{lem:mainintegrationlemmageneral} to the normal distribution of interest.
\begin{theorem}\label{thm:mainintegral}
Define
\begin{align}
\Gamma_1&\coloneqq - ((\Sigma^{-1})_{AA})^{-1}(\Sigma^{-1})_{AB}-\Sigma_{AB}(\Sigma_{BB})^{-1}\\
\Gamma_2&\coloneqq ((\Sigma^{-1})_{AA})^{-1}(\Sigma^{-1})_{AB}\ ,
\end{align}
and set
\begin{align}
c(\xi,\Delta b)&\coloneqq \Gamma_1 \xi+\Gamma_2 \Delta b\qquad\textrm{ for }\qquad \xi\in\mathbb{R}^{n-k}\qquad\textrm{ and }\qquad b\in\mathbb{Z}^{n-k}\ .
\end{align}
Let $Y\sim \cN(0,\Sigma_*)$ be a centered normal random variable on $\mathbb{R}^k$ with covariance matrix~$\Sigma_*$
given by the Schur complement
\begin{align}
\Sigma_*&=\Sigma_{AA}-\Sigma_{AB}(\Sigma_{BB})^{-1}\Sigma_{BA}\
\end{align} of $\Sigma$. 
Let~$f_{X_B}$ denote the probability density function of~$X_B\sim \cN(0,\Sigma_{BB})$. Then the following holds: there is a family $\{\cR_b\}_{b\in\mathbb{Z}^{n-k}}$  of measurable subsets of $\mathbb{R}^{n-k}$ such that $\cR_b\subset \Delta b+\Pi_B\Sigma^{1/2}\cP_{\Sigma^{-1/2}}$ for each $b\in\mathbb{Z}^{n-k}$, 
\begin{align}\label{eq:pepsboundgeneral1}
\sum_{b\in\mathbb{Z}^{n-k}}\int_{\cR_{b}} f_{X_B}(\xi)d\xi \leq 1\ ,
\end{align}
and the decoding success probability $p_S(\epsilon)$ is bounded by 
\begin{align}
p_S(\epsilon)&\leq \sum_{b\in\mathbb{Z}^{n-k}} \int_{\cR_b} f_{X_B}(\xi)\cdot\Pr\left[Y\in \cB_\epsilon(c(\xi,\Delta b))\right] d\xi\ .\label{eq:pepsboundgeneral2}
\end{align}
\end{theorem}

\begin{proof}
By Theorem~\ref{thm:mainvoronoireformulation} we have 
\begin{align}
p_S(\epsilon)&=\Pr\left[X\in (\cB_\epsilon(0)\times \Delta \mathbb{Z}^{n-k})+\Sigma^{1/2}\cV_{\Sigma^{-1/2}}\right]
\end{align}
where $X\sim\cN(0,\Sigma)$. 
With $\cV_{\Sigma^{-1/2}}\subset \cP_{\Sigma^{-1/2}}$ we obtain
 the upper bound
 \begin{align}
 p_S(\epsilon)&\leq \int_{\bigcup_{b\in\mathbb{Z}^{n-k}} \cS_b(\epsilon)}f_{X}(x)dx\qquad\textrm{ where }\qquad \cS_b=\binom{0}{\Delta b}+\Sigma^{1/2}\cP_{\Sigma^{-1/2}}\textrm{ for } b\in\mathbb{Z}^{n-k}\ .
  \end{align}
  With Lemma~\ref{lem:mainintegrationlemmageneral} we conclude that
  \begin{align}
   p_S(\epsilon)&\leq \sum_{b\in\mathbb{Z}^{n-k}}\int_{\xi\in\Pi_B\icS_b}f_{X_B}(\xi) I(b,\xi) d\xi\ , 
     \label{eq:PSepsilonsum}
     \end{align}
    where (see~\eqref{eq:Ibxidef})
     \begin{align}
     I(b,\xi)&\coloneqq \int_{h\in\cB_\epsilon(0)} f_{X_A|X_B=\xi}\left(\hat{\Lambda}(\xi-\Delta b)+h\right)dh\qquad\textrm{ for }(\xi,b)\in\mathbb{R}^{n-k}\times\mathbb{Z}^{n-k}\ ,
              \end{align}
with $\hat{\Lambda}=- ((\Sigma^{-1})_{AA})^{-1}(\Sigma^{-1})_{AB}$  (see Eq.~\eqref{eq:hatlambdadefinition}). 
According to~\eqref{eq:Ibxidef},
\begin{align}
I(b,\xi)=\Pr\left[\left. X_A\in \cB_\epsilon\left(\hat{\Lambda}(\xi-\Delta b)\right)\right|X_B=\xi\right]\ 
\end{align}
is the probability mass of an $\epsilon$-ball centered at~$\hat{\Lambda}(\xi-\Delta b)$ with respect to the conditional distribution $f_{X_A|X_B=x_B}$. Since we are considering $X\sim\cN(0,\Sigma)$,  the latter is normal with covariance matrix given by the Schur complement~$\Sigma_*$ of $\Sigma$ and centered at
\begin{align}
m_*&=\Sigma_{AB}(\Sigma_{BB})^{-1}\xi\ .
\end{align} 
Translating, i.e., considering $Y\coloneqq X_A-m_*$, we can write this as
\begin{align}
I(b,\xi)=\Pr\left[Y\in \cB_\epsilon\left(\hat{\Lambda}(\xi-\Delta b)-m_*\right)\right]\qquad\textrm{ where }\qquad Y\sim \cN(0,\Sigma_*)\ .\label{eq:ibxicomputed}
\end{align}
Inserting~\eqref{eq:ibxicomputed} into~\eqref{eq:PSepsilonsum} and setting~$\cR_b:=\Pi_B\icS_b$ for $b\in\mathbb{Z}^{n-k}$ gives Eq.~\eqref{eq:pepsboundgeneral2} since for $X=(X_A,X_B)\sim\cN(0,\Sigma)$, we have $X_B\sim \cN(0,\Sigma_{BB})$. Eventually, Eq.~\eqref{eq:pepsboundgeneral1} follows from Lemma~\ref{lem:mainintegrationlemmageneral}.
\end{proof}

\begin{corollary}\label{cor:upperboundpsepsilon}
Let $\lambda_{\min}(M)$ denote the smallest eigenvalue of a positive definite matrix~$M$. Let $\mu_Z$ denote the probability measure of a centered normal Gaussian distribution, $Z\sim\cN(0,I_k)$, on~$\mathbb{R}^k$. 
Then
\begin{align}
p_S(\epsilon) &\leq  \mu_Z\left(\cB_{\epsilon/\sqrt{\lambda_{\min}(\Sigma_*)}}(0)\right)\ .\label{eq:firstclaimpsepsilon}
\end{align}
In particular,
\begin{align}
p_S(\epsilon) &\leq  \mu_Z\left(\cB_{\epsilon/\sqrt{\lambda_{\min}(\Sigma)}}(0)\right)\ .\label{eq:secondclaimpsepsilon}
\end{align}
\end{corollary}
\begin{proof}
Let $Y\sim \cN(0,\Sigma_*)$ and let $c\in\mathbb{R}^{k}$ be arbitrary. Then
\begin{align}
\Pr\left[Y\in\cB_\epsilon(c)\right]&=\Pr\left[Z\in \Sigma^{-1/2}_*\cB_\epsilon(c)\right]\ .
\end{align}
With
$\Sigma^{-1/2}_*\cB_\epsilon(c)=\Sigma_*^{-1/2}\left(c+\cB_\epsilon(0)\right)=\Sigma_*^{-1/2}c+\Sigma_*^{-1/2}\cB_\epsilon(0)$
and $\Sigma_*^{-1/2}\cB_\epsilon(0)\subset \cB_{\epsilon\cdot \lambda_{\max}(\Sigma_*^{-1/2})}(0)=  \cB_{\epsilon/\sqrt{\lambda_{\min}(\Sigma_*)}}(0)$ we obtain
\begin{align}
\Pr\left[Y\in\cB_\epsilon(c)\right]&\leq \mu_Z\left(\cB_{\epsilon_*}(c_*)\right)\ ,
\end{align}
with $\epsilon_*=\epsilon/\sqrt{\lambda_{\min}(\Sigma_*)}$ and $c_*=\Sigma_*^{-1/2}c$.  We have thus related $\Pr\left[Y\in\cB_\epsilon(c)\right]$ to the  probability mass of a $k$-dimensional $\epsilon_*$-ball centered at~$c_*$ under the canonical Gaussian measure~$\mu_Z$. The latter is maximal when the ball is centered at the origin, i.e.,
\begin{align}
\mu_Z\left(\cB_{\epsilon_*}(c_*)\right)&\leq \mu_Z\left(\cB_{\epsilon_*}(0)\right)\qquad\textrm{ for all }c_*\in\mathbb{R}^{k}\ ,
\end{align}
according to Anderson's inequality~\cite{anderson55}. (More generally, the latter implies that for a convex set~$C\in\mathbb{R}^k$ which is symmetric about the origin $\mu_{Z}(C)\geq \mu_{Z}(C+a)$ for any $a\in\mathbb{R}^k$.) We have thus shown that
\begin{align}
\Pr\left[Y\in\cB_\epsilon(c)\right]& \leq \mu_Z\left(\cB_{\epsilon/\sqrt{\lambda_{\min}(\Sigma_*)}}(0)\right)\qquad\textrm{ for alll }c\in\mathbb{R}^k\ .
\end{align}
Inserting this into Eq.~\eqref{eq:pepsboundgeneral2} of Theorem~\ref{thm:mainintegral}, and subsequently applying Eq.~\eqref{eq:pepsboundgeneral1} of the same theorem, we obtain~\eqref{eq:firstclaimpsepsilon}.

The claim~\eqref{eq:secondclaimpsepsilon} follows immediately from~\eqref{eq:firstclaimpsepsilon}. This is because 
 the  eigenvalues of the Schur complement $\Sigma_*$ can be related to that of the positive definite matrix~$\Sigma$  using the inequality
$\lambda_{\min}(\Sigma_*)\geq \lambda_{\min}(\Sigma)$, see~\cite[Theorem 5]{smith92}. For completeness, 
we give a proof of this inequality in Appendix~\ref{app:proofschurcomplement}.
\end{proof}

\section*{Acknowledgements}
RK acknowledges support by  the DFG cluster of excellence 2111 (Munich Center for Quantum Science
and Technology). 
RK and LH acknowledge support by the German Federal Ministry of Education through the  program
Photonics Research Germany, contract no. 13N14776 (QCDA-QuantERA). 
The research of LH was partially supported by the Swiss National Science Foundation (Grant No. P2EZP2-188093).

\appendix
\section{Quadratic constraints for the degenerate Voronoi cell\label{app:quadraticform}}
\begin{lemma}\label{lem:quadraticformb}
Consider the set
\begin{align}
\cE_{b}\coloneqq \left\{z\in\mathbb{R}^{n}\ \left|\ \left\langle\Lambda\binom{h}{\Delta b},z\right\rangle \leq \frac{1}{2} \left\|\Lambda \binom{h}{\Delta b}\right\|^2\textrm{ for all }h\in\mathbb{R}^k\right.\right\}\qquad\textrm{ with } b\in\mathbb{Z}^{n-k}\backslash \{0\}\ .
\end{align}
Define $w_b(z)\coloneqq (\Lambda^2)_{AB}\Delta b-\Pi_A\Lambda z$, $\gamma_b(z)\coloneqq \|\Lambda\binom{0}{\Delta b}\|^2-2\langle z,\Lambda\binom{0}{\Delta b}\rangle$, and $\Omega:=(\Lambda_{AA})^2+\Lambda_{AB}\Lambda_{BA}$. Let $q_b:\mathbb{R}^{n}\rightarrow\mathbb{R}$ be the quadratic form
\begin{align}
q_b(z)\coloneqq \gamma_b(z)-\langle\Omega^{-1}w_b(z),w_b(z)\rangle\ .
\end{align}
Then 
\begin{align}
\cE_b=\left\{z\in\mathbb{R}^n\ |\ q_b(z)\geq 0 \right\}
\end{align}
\end{lemma}

\begin{proof}
We have $z\in\cE_b$ if and only if 
\begin{align}
\alpha(h)+\beta_{b}(h,z)+\gamma_{b}(z)\geq 0\textrm{ for every }h\in\mathbb{R}^k
\end{align}
where 
\begin{align}
\alpha(h)&\coloneqq \left\|\Lambda\binom{h}{0}\right\|^2\ ,\\
\beta_{b}(h,z)&\coloneqq 2\left\langle \Lambda\binom{h}{0},\Lambda\binom{0}{\Delta b}-z\right\rangle\ ,\\
\gamma_{b}(z)&\coloneqq \left\|\Lambda\binom{0}{\Delta b}\right\|^2-2\left\langle z,\Lambda\binom{0}{\Delta b}\right\rangle\ .
\end{align}
Observe that
\begin{align}
\begin{matrix}
\alpha(h)&=&\langle h,\Omega h\rangle\qquad&\textrm{ with }\qquad &\Omega&\coloneqq &(\Lambda_{AA})^2+\Lambda_{AB}\Lambda_{BA}\\
\beta_{b}(h,z)&=&2\langle h,w_b(z)\rangle\qquad&\textrm{ with }\qquad &w_b(z)&\coloneqq &(\Lambda^2)_{AB}\Delta b-\Pi_A\Lambda z\ .
\end{matrix}
\end{align}
Since $\Lambda^2$ is positive definite, the quadratic form $Q(h):=\alpha(h)+\beta_{b}(h,z)+\gamma_{b}(z)$ is minimal
when $\nabla_h Q(h)=0$. This is the case for $h_*:=-\Omega^{-1}w_b(z)$. The corresponding value is
\begin{align}
Q(h_*)&=\gamma_{b}(z)-\langle \Omega^{-1}w_b(z),w_b(z)\rangle\ .
\end{align}
Since $z\in\cE_b$ if and only if $Q(h_*)\geq 0$, the claim follows.
\end{proof}

\section{A lower bound on the minimal eigenvalue of the Schur complement\label{app:proofschurcomplement}}

In this appendix, we give a proof of the following statement for completeness. 
\begin{lemma}\label{lem:mineigenvaluesschurcompl}
Let $\Sigma=\begin{pmatrix}
\Sigma_{AA} & \Sigma_{AB}\\
\Sigma_{BA} & \Sigma_{BB}
\end{pmatrix}$
be positive definite,
 with $\Sigma_{AA}\in\mathsf{Mat}_{(n-r)\times (n-r)}(\mathbb{R})$,
$\Sigma_{AB}=\Sigma_{BA}^T\in\mathsf{Mat}_{(n-r)\times r}(\mathbb{R})$, and $\Sigma_{BB}\in \mathsf{Mat}_{r\times r}(\mathbb{R})$. Let
\begin{align}
\Sigma_*&=\Sigma_{AA}-\Sigma_{AB}(\Sigma_{BB})^{-1}\Sigma_{BA}\ .
\end{align}
denote the Schur complement. Then
\begin{align}
\lambda_{\min}(\Sigma_*)\geq \lambda_{\min}(\Sigma)\ .
\end{align}
\end{lemma}
We refer to~\cite[Theorem 5]{smith92} for a more general statement showing that the eigenvalues of the Schur complement of a semidefinite Hermitian matrix interlace the eigenvalues of the matrix itself. The proof relies on the Weyl inequalities. Here we only need  the following statement. 
Let
\begin{align}
\lambda_{\min}(A)=\lambda_1(A)\leq \lambda_2(A)\leq \cdots\leq \lambda_{n-1}(A)\leq \lambda_n(A)=\lambda_{\max}(A)
\end{align}
denote the ordered eigenvalues of a Hermitian $n\times n$-matrix~$A$.  Then the following holds:
\begin{lemma} (see~\cite[Corollary 4.3.5]{hornjohnson2012})
Let $C,D$ be Hermitian $n\times n$-matrices. Assume that $D$ has rank at most~$r$. Then
\begin{align}
\lambda_k(C+D) \leq \lambda_{k+r}(C)\qquad\textrm{ for all }k=1,2,\ldots,n-2r\ .\label{eq:weyllike}
\end{align}
\end{lemma}

\begin{proof}[Proof of Lemma~\ref{lem:mineigenvaluesschurcompl}]
For brevity, let us write $X:=\Sigma_{AA}$, $Y=\Sigma_{AB}$ and $Z:=\Sigma_{BB}$. Note that both $X$ and $Z$ are positive definite as principal submatrices  of the positive definite matrix~$\Sigma=\begin{pmatrix}
X & Y\\
Y^T & Z
\end{pmatrix}$, and $\Sigma_*=X-YZ^{-1}Y^T$. We have
\begin{align}
\begin{pmatrix}
X & Y\\
Y^T & Z
\end{pmatrix}&=\begin{pmatrix}
X-YZ^{-1}Y^T & 0 \\
0 & 0
\end{pmatrix}+\begin{pmatrix}
YZ^{-1}Y^T & Y\\
Y^T & Z
\end{pmatrix}\ .
\end{align}
That is,
\begin{align}
\Sigma &=C+D\label{eq:inequalitySigmaschurcomplement}
\end{align}
where $C:=\begin{pmatrix}
\Sigma_* & 0\\
0 & 0
\end{pmatrix}$ and where the matrix
\begin{align}
D:=\begin{pmatrix}
YZ^{-1}Y^T & Y\\
Y^T & Z
\end{pmatrix}=\begin{pmatrix}
YZ^{-1/2}\\
Z^{1/2}
\end{pmatrix} \left( Z^{-1/2}Y^T\quad Z^{1/2}\right)
\end{align}
is positive semidefinite and of rank at most~$r$. 

Applying the inequality~\eqref{eq:weyllike} with $k=1$ to $(C,D)$ and inserting~\eqref{eq:inequalitySigmaschurcomplement} thus gives
\begin{align}
\lambda_{\min}(\Sigma)&\leq \lambda_{r+1}(C)\ .
\end{align}
This implies the claim since $\lambda_{r+1}(C)=\lambda_{\min}(\Sigma_*)$.
\end{proof}

\bibliographystyle{plain}
\bibliography{q}

\end{document}